\documentclass[11pt]{article}
\usepackage{amssymb,mathrsfs,verbatim,graphicx,rotate}
\usepackage{amsmath,amsthm}
\usepackage{graphics}
\usepackage{color}
\usepackage{float}
\usepackage{lscape}
\usepackage{setspace}
\usepackage{subfigure}
\usepackage{natbib}
\usepackage{multirow}
\usepackage{enumerate}
\usepackage{array}
\usepackage[hidelinks]{hyperref}
\usepackage{prodint}
\usepackage{thmtools}
\usepackage{changepage}
\usepackage[shortlabels]{enumitem}

\usepackage[mathscr]{euscript}
\usepackage[bottom,flushmargin,hang,multiple]{footmisc}
\usepackage[letterpaper, margin=1in]{geometry}

\declaretheorem[name=Theorem]{thm}

\declaretheorem[name=Lemma]{lemma}

\DeclareMathOperator*{\inprob}{\stackrel{P}{\longrightarrow}}

\DeclareMathOperator*{\indist}{\stackrel{d}{\longrightarrow}}

\newcommand{\bounded}{O_P}
\newcommand{\fasterthan}{o_P}

\newcommand{\boundeddet}{O}

\DeclareMathOperator*{\logit}{logit}

\newcommand{\s}[1]{\mathcal{#1}}
\renewcommand{\d}[1]{\mathbb{#1}}

\newcommand{\n}[1]{\mathrm{#1}}

\DeclareMathOperator*{\argmin}{argmin}

\newcommand\independent{\protect\mathpalette{\protect\independenT}{\perp}}
\def\independenT#1#2{\mathrel{\rlap{$#1#2$}\mkern2mu{#1#2}}}

\newcommand{\m}{\text{\hspace{.02in}--}}

\makeatletter
\def\namedlabel#1#2{\begingroup
    #2%
    \def\@currentlabel{#2}%
    \phantomsection\label{#1}\endgroup
}
\makeatother
\newcolumntype{x}[1]{%
>{\centering\hspace{0pt}}p{#1}}%

\begin{document}

\title{Inference for treatment-specific survival curves\\ using machine learning}
\author{Ted Westling \\ Department of Mathematics and Statistics \\ University of Massachusetts Amherst \\ twestling@math.umass.edu \and Alex Luedtke \\ Department of Statistics\\ University of Washington \\ aluedtke@uw.edu \and Peter Gilbert \\ Vaccine and Infectious Disease Division \\ Fred Hutchinson Cancer Research Center \\pgilbert@scharp.org \and Marco Carone\\ Department of Biostatistics\\ University of Washington \\mcarone@uw.edu}
\maketitle

\begin{abstract}
In the absence of data from a randomized trial, researchers often aim to use observational data to draw causal inference about the effect of a treatment on a time-to-event outcome. In this context, interest often focuses on the treatment-specific survival curves; that is, the survival curves were the entire population under study to be assigned to receive the treatment or not. Under certain causal conditions, including that all confounders of the treatment-outcome relationship are observed, the treatment-specific survival can be identified with a covariate-adjusted survival function. Several estimators of this function have been proposed, including estimators based on outcome regression, inverse probability weighting, and doubly robust estimators. In this article, we propose a new cross-fitted doubly-robust estimator that incorporates data-adaptive (e.g.\ machine learning) estimators of the conditional survival functions. We establish conditions on the nuisance estimators under which our estimator is consistent and asymptotically linear, both pointwise and uniformly in time. We also propose a novel ensemble learner for combining multiple candidate estimators of the conditional survival estimators. Notably, our methods and results accommodate events occurring in discrete or continuous time (or both). We investigate the practical performance of our methods using numerical studies and an application to the effect of a surgical treatment to prevent metastases of parotid carcinoma on mortality.
\end{abstract}

\doublespacing

\clearpage

\section{Introduction}\label{intro}

\subsection{Motivation and literature review}

The gold standard for assessing the causal effect of a binary treatment on a time-to-event outcome is a randomized control trial  in which participants are randomly assigned to treatment or control and followed over time. The effect of treatment may then be assessed by comparing the fraction of participants who experience the event by the end of the study in the treatment and control arms. However, some participants' outcomes may be unknown for various reasons, such as dropping out of the study or moving away from the study site, which is known as right-censoring of the event time. If the time of right-censoring is independent of the event time conditional on treatment arm, then contrasts of the stratified Kaplan-Meier estimators can be used to assess the treatment effect \citep{kaplan1958}.

Randomizing treatment status is often infeasible or unethical, or preliminary evidence may be needed to justify the cost of conducting a randomized trial. In such cases, researchers may turn to observational data --- obtained, for example, from cohort studies, registries, or electronic medical records. In such contexts, the treatment or exposure is not randomized, but instead assigned or selected according to an unknown mechanism. Assessing the causal effect of a treatment on a time-to-event outcome with observational data is challenging due to confounding of the treatment-outcome relationship. When there are confounding variables that affect the treatment selection or assignment process and also impact the outcome, simple approaches such as contrasts of stratified Kaplan-Meier estimators typically cannot be interpreted causally. Any observed differences (or lack thereof) in the outcome between those who received treatment and those who did not may be due to the confounding variables rather than the treatment. Even if the treatment is randomized, dependence of the event and censoring times can also render the Kaplan-Meier estimator inconsistent and resulting inference invalid.

If the recorded covariates are rich enough to de-confound the treatment-outcome, treatment-censoring, and outcome-censoring relationships, then a causal effect may still be recovered, and there are a variety of existing methods for doing so. The most common approach consists of fitting a time-to-event regression model, such as a Cox proportional hazards model \citep{cox1972regression}. If the Cox model holds, then the exponentiated regression coefficient corresponding to treatment can be interpreted as a conditional hazard ratio comparing treated and control patients. Alternatively, any time-to-event regression model can be marginalized using the G-formula to obtain estimated treatment-specific survival curves corresponding to the hypothetical scenarios in which all patients are assigned to treatment or control \citep{makuch1982adjusted}. However, if the model is misspecified, the treatment regression coefficient estimator may not be consistent for a scientifically meaningful quantity, and resulting marginalized survival curves will typically be inconsistent. 
As an alternative to outcome regression models, inverse probability weighting may be used \citep{cole2004adjusted}. However, their consistency hinges on consistent estimation of both the treatment assignment mechanism and the censoring distribution. Finally, if the event time of interest is known to take values on a finite grid of time-points, then methods for longitudinal data can be used --- see, e.g.\ \cite{rotnitzky2012improved} and references therein. However, using a discrete-time approximation for an event truly occurring in continuous time generally yields inconsistent estimators \citep{guerra2020discretization}.

Doubly-robust estimators combine regression and weighting estimators in such a way that the bias of the resulting estimator is a product of the biases of the outcome regression and weighting function estimators. As a result, doubly-robust estimators are consistent if \emph{either} the outcome regression function \emph{or} the weighting function estimators are consistent. Furthermore, doubly-robust estimators can converge in distribution to a normal limit at the parametric rate even when flexible (e.g.,\ machine learning) procedures are used to construct the outcome regression and weighting function estimators.

Several doubly-robust estimators of treatment-specific survival curves in continuous time have been proposed. \cite{zeng2004marginal} proposed an estimator that is consistent as long as either the conditional time-to-event or censoring distributions follows a Cox model. \cite{zhang2012dr} proposed an estimator that is consistent as long as either the conditional time-to-event distribution follows a Cox model or the treatment assignment mechanism follows a logistic regression model. Finally,  \cite{hubbard2000survival} and \cite{bai2013dr} proposed doubly-robust estimators based on semiparametric efficiency theory, and suggested using common semiparametric regression models to estimate the conditional time-to-event and censoring distributions.

\subsection{Contribution and organization of the article}

To the best of our knowledge, the use of machine learning techniques for doubly-robust estimation of treatment-specific survival curves (and contrasts thereof) permitting events occurring in continuous time has not yet been studied. 
In this article, we  fill this gap in the literature. Specifically, in this paper, we make the following contributions:
\begin{enumerate}[(1)]
\item we derive a nonparametric identification of the treatment-specific counterfactual survival;
\item we propose a novel cross-fitted one-step estimator of the treatment-specific survival curve that permits the use of machine learning for nuisance function estimation;
\item we provide general conditions under which our estimator is (uniformly) consistent and (uniformly) asymptotically linear;
\item we propose methods for pointwise and uniform inference;
\item we propose a novel ensemble learner for combining multiple candidate estimators of the conditional survival functions.
\end{enumerate}In addition, we conduct a numerical study and apply our results to assess the effect of  elective neck dissection on all-cause mortality using an observational cohort of patients with parotid carcinomas.

For ease of use, we have made the estimator and associated inferential procedures proposed here available through the \texttt{R} package \texttt{CFsurvival} (https://github.com/tedwestling/CFsurvival), and we have implemented the method proposed in Section~\ref{sec:nuisance} for estimating conditional survival functions in the \texttt{R} package \texttt{survSuperLearner} (https://github.com/tedwestling/survSuperLearner).


\section{Statistical setting and parameters of interest}

\subsection{Ideal and observed data structures}

We now define the ideal data structure we consider in temporal order. As we discuss below, we only observe a coarsening of  this ideal data structure. First, we record a vector $W$ of baseline covariates taking values in $\s{W} \subseteq \d{R}^d$. After recording $W$, but prior to time $t=0$, we observe a binary exposure $A\in\{0,1\}$. Adopting the Neyman-Rubin potential outcomes framework \citep{neyman1923, rubin1974estimating}, we let $T(a)$ be the event time of interest under assignment to exposure $A = a$. We assume that for $a \in \{0,1\}$, $T(a)$ takes values in $(0, \infty]$. Since we assume that $T(a) > 0$, all patients start the study without having experienced the event of interest, and since we allow $T(a) = \infty$, some patients may never experience the event. We then let $C(a)$ be a right-censoring time under assignment to exposure $A=a$, and we assume that $C(a) \in [0,\infty]$. Since we allow $C(a) = 0$, patients may be censored immediately if, for instance, a patient is lost to follow-up just after the exposure $A$ is recorded.  We define $O_F := (W, A, T(0), T(1), C(0), C(1))$ to be the ideal data unit, and denote by $P_{0,F}$ the distribution of $O_F$. Throughout, we assume that each patient's potential event and censoring times are independent of all other patients' exposures. 

We now describe the coarsened version of $O_F$ actually observed. We denote by $T := T(A)$ and $C:=C(A)$ the event and censoring times corresponding to the exposure received. We assume that the right-censored time $Y:=\min\{T, C\}$ and the event indicator $\Delta := I(T \leq C)$ are observed for each patient. Thus, the available data consist of $n$ independent and identically distributed observations $O_1, O_2,\dotsc, O_n$ of the observed data unit $O := \left(W, A, Y, \Delta\right)$. We denote by $P_0$ the distribution of the observed data unit, as induced by the distribution $P_{0,F}$ of the ideal data unit.

Throughout, we denote summaries of $P_0$ with the subscript $0$, e.g.,\ $E_0[ f(O)] := E_{P_0}[f(O)]$, and summaries of $P_{0,F}$ with subscript $0,F$. In cases where $f$ is a random function, the expectation $E_0[f(O)]$  should be understood as being taken with respect to the distribution of the random unit $O$, but not the function $f$. In addition, we let $a \wedge b$ denote $ \min\{a, b\}$,  $\d{P}_n$ be the empirical distribution corresponding to $O_1, O_2,\dotsc, O_n$, and $Pf := \int f(o) \, dP(o)$ for any probability measure $P$ and $P$-measurable function $f$. Finally, throughout, we use the convention $0/0 := 1$.

\subsection{Causal parameter of interest and identification}

In this article, we are interested in the causal survival curves $t\mapsto \theta_{0,F}(t,a) := P_{0,F}(T(a) > t)$ for $a \in \{0,1\}$ and $t \in [0,\tau]$ for some positive $\tau<\infty$. Thus defined, $\theta_{0,F}(t,0)$ represents the population probability that a patient would experience the event later than time $t$ if, contrary to fact, all patients were assigned to receive the control exposure $(A = 0)$, and $\theta_{0,F}(t,1)$ represents the same if, contrary to fact, all patients were assigned to receive the treatment under study $(A = 1)$. In addition to the exposure-specific survival curves $t \mapsto \theta_{0,F}(t,0)$ and $t \mapsto \theta_{0,F}(t,1)$, we are interested in survival contrasts including the survival difference $t \mapsto \theta_{0,F}(t,1) - \theta_{0,F}(t,0)$, survival ratio $t \mapsto \theta_{0,F}(t,1) / \theta_{0,F}(t,0)$, and risk ratio $t \mapsto [1-\theta_{0,F}(t,1)] / [1-\theta_{0,F}(t,0)]$ functions.

Under certain conditions, we can identify the causal parameter $\theta_{0,F}(t,a)$ in terms of the distribution $P_0$ of the observed data unit. If the conditions (i) $T(a) \independent A$, (ii) $C(a) \independent A$, (iii) $T(a) \independent C(a) \,|\, A$, (iv) $P_0(A =a) > 0$, and (v) $P_{0,F}(C(a) \geq t) > 0$ hold, then  \[\theta_{0,F}(t, a) = \bar{S}_0(t \,|\, a) := \prodi_{(0,t]} \{ 1- \Lambda_0(du \,|\, a)\}\ ,\] where $\Lambda_0(t \,|\, a) := \int_0^t F_{0,1}(du \,|\, a) / R_0(u\,|\,a)$ with $F_{0,1}(u \,|\, a) := P_0(Y \leq u, \Delta = 1 \,|\, A=a)$ and $R_0(u \,|\, a) = P_0(Y \geq u \,|\, A =a)$, and where $\prodi$ denotes the Riemann-Stieltjes product integral \citep{gill1990product}.  These are the standard parameters estimated by the treatment group-specific Kaplan-Meier estimators in the context of randomized control trials. Indeed, when $A$ is randomized, conditions (i), (ii) and (iv) are automatically satisfied. If patients are uncensored at time $\tau$ with probability one, then conditions (iii) and (v) are satisfied as well. If some patients are censored before time $\tau$, then conditions (iii) and (v) represent the \emph{independent censoring} assumption, which we note may fail to hold even in the context of a randomized trial.

We note that the product integral reduces to a product $\prod_{t_j \leq t} \{1 - \lambda_0(t_j \, | \, a)\}$ in the case of events occurring in discrete time, and reduces to the exponential $\exp\{-\Lambda_0(t \,|\, a)\}$ in the case of events occurring in continuous time. Here and throughout, we use the product integral in order to permit either of these more familiar cases, or a time-to-event distribution including both discrete and continuous components.

In observational studies, where the exposure $A$ is not randomized, the exposure-outcome, exposure-censoring, or outcome-censoring relationships are often confounded. In such cases, the parameter $\bar{S}_0(t \,|\, a)$ of the observed data distribution  may no longer coincide with $\theta_{0,F}(t,a)$. However, if $W$ is a sufficiently rich set of recorded pre-exposure covariates,  then $\theta_{0,F}(t,a)$ may still be identified as a function of the observed data distribution. Specifically, we introduce the following identification conditions, which are specific to the fixed values of $\tau \in (0, \infty)$ and $a \in \{0,1\}$:
\begin{description}[style=multiline,leftmargin=2cm, labelindent=.9cm]
	\item[\namedlabel{itm:t_exchang}{(A1)}] $T(a)I(T(a) \leq \tau) \independent A \,|\, W$;\vspace{-.075in}
	\item[\namedlabel{itm:c_exchang}{(A2)}] $C(a) I(C(a) \leq \tau) \independent A \,|\, W$;\vspace{-.075in}
	\item[\namedlabel{itm:tc_exchang}{(A3)}] $T(a) I(T(a) \leq \tau) \independent C(a)I(C(a) \leq \tau) \,|\, A=a, W$;\vspace{-.075in}
	\item[\namedlabel{itm:a_pos}{(A4)}] $P_0(A = a\,|\, W) > 0$ $P_0$-almost surely;\vspace{-.075in}
	\item[\namedlabel{itm:c_pos}{(A5)}] $P_{0,F}(C(a) \geq \tau \,|\, W) > 0$  $P_0$-almost surely.
\end{description}
Defining $F_{0,1}(t \,|\, a, w) := P_0(Y \leq t, \Delta = 1 \,|\, A = a, W = w)$, $R_0(t \,|\, a, w) := P_0(Y \geq t\,|\, A = a, W = w)$, and $\Lambda_0(t \,|\, a, w) := \int_0^t F_{0,1}(du \,|\, a, w) / R_0(u \,|\, a, w)$ for each $(t,a,w)$, we have the following identification result.
\begin{thm}\label{thm:ident}
If  conditions (A1)--(A5) hold for some $a \in \{0,1\}$ and $\tau \in (0, \infty)$, then $P_{0,F}(T(a) > t \,|\, W) = S_0(t \,|\, a, W)$ $P_0$-almost surely for all $t \in [0, \tau]$, where
$S_{0}(t \,|\, a, w) := \prodi_{(0,t]} \left\{ 1 - \Lambda_0(du \,|\, a, w) \right\}$, and so  $\theta_{0,F}(t, a) = E_0\left[ S_0(t \,|\, a, W)\right]$.
\end{thm}
Theorem~\ref{thm:ident} is a combination of the G-formula (also known as the backdoor or the regression standardization formula) from causal inference \citep{robins1986, gill2001} and the identification of a survival function in the context of dependent censoring \citep{beran1981censored, dabrowska1989uniform}.  Condition~\ref{itm:t_exchang} is a restricted form of exchangeability of the exposure-outcome relationship. The restriction to the event $T(a) \leq \tau$ allows aspects of the event mechanism occurring after time $\tau$ to depend on $A$, which is permitted if we only want to identify the survival probability up to time $\tau$. Conditions~\ref{itm:c_exchang} and~\ref{itm:tc_exchang} are analogously restricted forms of exchangeability of the exposure-censoring and outcome-censoring relationships, respectively. Notably, condition~\ref{itm:tc_exchang} permits that the event and censoring times be dependent, as long as they are conditionally independent given $A$ and $W$. Condition~\ref{itm:a_pos} is the usual positivity condition for the exposure, and condition~\ref{itm:c_pos}  requires that there is a positive probability of remaining uncensored in almost every stratum defined by $W$. Hence, if it is known that all patients will be censored with probability one by time $\tau$, then we can only estimate $\theta_{0,F}(t,a)$ for $t \leq \tau$. If it is known that patients within some strata of $W$ will be censored with probability one prior to $t_0 < \tau$, we can either estimate $\theta_{0,F}(t,a)$ only up to time $t_0$, or exclude  such patient subpopulations from the target population.


If conditions (A1)--(A5) hold for both $a =0$ and $a = 1$, then Theorem~\ref{thm:ident} implies that $\theta_{0,F}(t,a) = \theta_0(t,a)$ for every $t \in [0, \tau]$ and $a \in \{0,1\}$, where
\begin{equation} \theta_0(t, a) := E_0 [ S_0(t \,|\, a, W)]\ . \label{eq:obs_param}\end{equation}
This latter parameter is referred to as the \emph{G-computed probability} that the event $T$ occurs after time $t$ given that exposure $A$ is set to $a$. This parameter measures the survival probability under exposure $A =a$ while adjusting for potential confounding between the exposure and the event of interest and for dependence between the event and censoring times. The curves $\{ \theta_0(t, 0) : t \in [0, \tau]\}$ and $\{ \theta_0(t, 1) : t \in [0, \tau]\}$, and contrasts thereof, are the  observed-data statistical parameters we focus on.

Even when conditions (A1)--(A5) do not strictly hold, so that a causal interpretation of $\theta_0(t,a)$ may not be appropriate, $\theta_0(t,a)$ may still be of greater scientific interest than the unadjusted survival probability $\bar{S}_0(t \,|\, a)$. The first reason for this is that $\theta_0$ allows for the adjustment of covariates related to both $A$ and $T$. As a result, $\theta_0$ can always be interpreted as the average probability that $T$ exceeds $t$ in a hypothetical population of patients with $A = a$ but with a distribution of the covariate vector $W$ identical to that in the target population. The second reason is that adjusting for $W$ allows the relaxation of the marginal independent censoring assumption $T(a) \independent C(a) \,|\, A$ to a conditional independent censoring assumption $T(a) \independent C(a) \,|\, A, W$. This relaxation can be important in contexts where the event and censoring times may be dependent, but the recorded covariate vector $W$ at least partly explains the dependence between them.



\section{Estimation}\label{survival}

\subsection{Efficiency calculations}

In this section, we describe the proposed methodology for nonparametric efficient estimation of the treatment-specific G-computed survival function $\{\theta_0(t, a) : t \in [0, \tau]\}$. The definition of our estimator involves  three steps. First, we characterize the behavior of nonparametric efficient estimators of the estimand of interest; notably, this is done by deriving the nonparametric efficient influence function (EIF) of $\theta_0(t,a)$ for each $t$ and $a$. Second, we use the explicit form of the EIF obtained to construct an efficient estimator of $\theta_0(t,a)$ for each $t$ and $a$. Finally, we describe a simple procedure to ensure that the resulting survival curves are monotone.

We define $\pi_0(a \,|\, w) := P_0(A = a \,|\, W = w)$. Below, we make use of the fact that, for any $(a,w)$ such that $S_0(u- \,|\, a, w) > 0$, $P_{0,F}(C \geq u \,|\, A = a, W =w)$ can be identified as a mapping of the distribution of the observed data under conditions (A1)--(A5) as
\[ P_{0,F}(C \geq u \,|\, A = a, W =w)= G_0(u\,|\, a, w) := \Prodi_{[0,u)} \left\{1 - H_0(ds \,|\, a, w) \right\},\]
where we define $H_0(u \,|\, a, w) := \int_{[0,u]} \left\{\frac{S_0(s\m \,|\, a, w)}{S_0(s \,|\, a, w)}\right\}\frac{F_{0,0}(ds \,|\, a, w)}{R_0(s\,|\, a, w)}$ and $F_{0,0}(u \,|\, a, w) := P_0(Y\leq u,\Delta = 0 \,|\, A = a, W = w)$. We emphasize that $G_0$ is defined as the left-continuous conditional survival function of $C$, whereas $S_0$ is defined as the right-continuous conditional survival function of $T$.

In Theorem~\ref{thm:eif}, we present the nonparametric efficient influence function of $\theta_0(t, a_0)$, where we use $a_0$ rather than $a$ to denote the exposure value of interest in order not to confuse values of the random variable $A$ with the specific $a_0$ at which we want to evaluate $\theta_0$. 
\begin{thm}\label{thm:eif}
If there exists $\eta > 0$ such that $\min\{\pi_0(a_0 \,|\, w),G_0(t \,|\, a_0, w)\} \geq \eta$ for $P_0$-almost every $w$ such that $S_0(t\,|\, a_0, w) > 0$, then $\theta_{0}(t, a_0)$ is a pathwise differentiable parameter in a nonparametric model with efficient influence function  $\phi_{0,t, a_0}^*:= \phi_{0,t,a_0} -  \theta_{0}(t, a_0)$, where $\phi_{0,t, a_0}(y, \delta, a, w)$ equals
\begin{align*}
S_0(t \,|\, a_0, w)\left[ 1 - \frac{I(a = a_0)}{ \pi_0(a \,|\, w)}\left\{ \frac{I(y \leq t, \delta = 1)}{S_0( y \,|\, a, w) G_0(y \,|\, a,  w)} - \int_{0}^{t \wedge y}\frac{\Lambda_0(du \,|\, a, w)}{S_0(u \,|\, a, w)G_0(u \,|\, a, w)} \right\} \right].
\end{align*}\end{thm} We note that \cite{hubbard2000survival} and \cite{bai2013dr} also derived the efficient influence function of $\theta_0(t,a_0)$ in a nonparametric model, but the form of the influence function presented in Theorem~\ref{thm:eif} is somewhat simpler than the forms previously published because it is parametrized in terms of variation independent nuisance functions.
We also note that if there is no covariate vector $W$ to adjust for, so that $\theta_0(t, a) =  \bar{S}_0(t \,|\, a)$ is the unadjusted conditional survival function, then $\phi_{0,t,a_0}^*$ reduces to the influence function of the stratified Kaplan-Meier estimator \citep{reid1980influence}.

\subsection{Cross-fitted one-step estimator}

The efficient influence function $\phi_{0,t, a_0}^*$ involves three variation-independent nuisance functions: $S_0$, $G_0$ and $\pi_0$. We  discuss estimation of these functions in Section~\ref{sec:nuisance}. We note that $\Lambda_0$ and $S_0$ are in one-to-one correspondence with one another, so estimating $S_0$ gives an estimator of $\Lambda_0$ and vice-versa. Given estimators $S_n$, $G_n$ and $\pi_n$ of $S_0$, $G_0$ and $\pi_0$, respectively, there are multiple possible asymptotically linear and efficient estimators of $\theta_{0}(t, a)$. Denoting by $\phi_{n,t,a}$ the function $\phi_{0,t,a}$ with $S_0$, $\Lambda_0$, $G_0$ and $\pi_0$ replaced by their respective estimators, the standard one-step estimator would be $\d{P}_n \phi_{n,t,a}$, which is also an estimating equations-based estimator in this case because the influence function is linear in $\theta_{0}(t, a)$. Asymptotic linearity of estimators of this type depend on nuisance estimators in two important ways. First, negligibility of a so-called \emph{second-order} remainder term requires that the nuisance parameters converge at fast enough rates to their true counterparts. Second, negligibility of an \emph{empirical process} remainder term can be guaranteed if the nuisance estimators fall in sufficiently small function classes with probability tending to one. In observational studies, researchers can rarely specify correct parametric models for nuisance parameters a priori, which motivates the use of data-adaptive  estimators. However, data-adaptive estimators --- especially the ensemble estimators that we propose in Section~\ref{sec:nuisance} --- typically fail to fall in small function classes. This poses a challenge in simultaneously achieving negligibility of these two remainder terms. Cross-fitting has been found to resolve this challenge by removing this constraint on the complexity of nuisance estimators (see, e.g.\ \citealp{bickel1982,robins2008higher, zheng2011cvtmle}, among many others). Therefore, we will use a cross-fitted version of the one-step estimator stated above, which we now define.

For a deterministic integer $K \in \{2, 3, \dotsc, \lfloor n/2\rfloor \}$, we randomly partition the indices $\{1,2, \dotsc, n\}$ into $K$ disjoint sets $\s{V}_{n,1}, \s{V}_{n,2}, \dotsc, \s{V}_{n,K}$ with cardinalities $n_1,n_2, \dotsc, n_K$. We require that these sets be of as close to equal sizes as possible, so that $|n_k - n/K| \leq 1$ for each $k$, and that the number of folds $K$ be bounded as $n$ grows. For each $k \in \{1, 2,\dotsc, K\}$, we define $\s{T}_{n,k} := \{ O_i : i \notin \s{V}_{n,k}\}$ as the \emph{training set} for fold $k$. We then define $S_{n,k}$, $G_{n,k}$, $\pi_{n,k}$ and $\Lambda_{n,k}$ as nuisance estimators estimated using only the observations from the training set $\s{T}_{n,k}$, and $\phi_{n,k, t,a}$ as the function $\phi_{0,t,a}$ in which these nuisance estimators have substituted their true counterparts. We then define the cross-fitted one-step estimator $\theta_n(t,a)$ of $\theta_0(t,a)$ pointwise as
\begin{align}
\theta_{n}(t,a) &:= \frac{1}{n} \sum_{k=1}^K \sum_{i\in \s{V}_{n,k}}\phi_{n,k, t,a}(O_i)\ . \label{eq:one_step}
\end{align}
We note that, once the nuisance functions are estimated, $\theta_n(t,a)$ can be efficiently  computed for many time-points $t$ because the same nuisance function estimators can be re-used for each $t$.

How the integral in $\phi_{n,k, t,a}$  is computed depends on the form of $S_{n,k}$. If $S_{n,k}$ is defined as a step function, then the integral reduces to a sum.  If $S_{n,k}$ is defined to be absolutely continuous, then the integral can be computed  with $\Lambda_{n,k}(du \,|\, \cdot) := \lambda_{n,k}(u \,|\, \cdot)\, du$.  Alternatively, the integral can be approximated as a sum on a finite grid.

\subsection{Enforcing monotonicity of the proposed estimator}

The function $t \mapsto \theta_0(t,a)$ is necessarily monotone non-increasing for each $a \in \{0,1\}$ and takes values in $[0,1]$. However, the proposed estimator $\theta_{n}(t,a)$ is generally neither guaranteed to lie in $[0,1]$ nor to be monotone in $t$ in any finite sample. 

We ensure that our final estimator satisfies the above parameter constraints as follows. First, we construct $\theta_n(t,a)$ as defined above for each $t \in \s{T}_n$, where $\s{T}_n$ is the set of unique values of $Y_1, Y_2,\dotsc, Y_n$. Second, for each $t \in \s{T}_n$ and $a \in \{0,1\}$, we define $\theta_n^+(t,a) = \theta_n(t,a)$ if $\theta_n(t,a) \in [0,1]$, $\theta_n^+(t,a)  = 1$ if $\theta_n(t,a) > 1$, and $\theta_n^+(t,a)  = 0$ if $\theta_n(t,a) < 0$. Next, for each $a \in \{0,1\}$, we define $\{\theta_n^\circ(t,a): t \in \s{T}_n\}$ as the projection of $\{\theta_{n}^+(t,a) : t \in \s{T}_n\}$ onto the space of non-increasing functions using isotonic regression. For any $t \in (0,\tau]$, we then define $\theta_n^\circ(t,a)$ as the evaluation of the right-continuous stepwise interpolation of $\{\theta_n^\circ(t,a): t \in \s{T}_n\}$. The projected estimator $\theta_n^\circ$ is guaranteed to be no farther from $\theta_0$ than $\theta_{n}$ in every finite sample, and if the true function is strictly decreasing, then the initial and projected estimators are asymptotically equivalent \citep{westling2020correcting}. Therefore, in what follows, we focus on providing large-sample results for $\theta_n$, since results for the isotonized estimator $\theta_n^\circ$ are identical in view  of the general results of \cite{westling2020correcting}.

\section{Large-sample properties}

\subsection{Consistency}

In this section, we study the large-sample properties of the proposed estimator. First, we provide conditions under which $\theta_n(t, a)$ is consistent for $\theta_0(t, a)$ for fixed $t$ and uniformly over $t$.

\begin{description}[style=multiline, labelindent=.9cm, leftmargin=2cm]
\item[\namedlabel{itm:conv}{(B1)}] There exist $\pi_{\infty}$, $G_{\infty}$ and $S_{\infty}$ such that:
\begin{enumerate}[(a), leftmargin=1cm, labelindent=1cm]
\item $\max_k E_{0} \left[ \frac{1}{\pi_{n,k}(a \, | \, W)} - \frac{1}{\pi_{\infty}(a \, | \,W)} \right]^2 \inprob 0$;
\item $\max_k E_{0} \left[  \sup_{u \in [0, t]} \left| \frac{1}{G_{n,k}(u \,|\, a, W)} - \frac{1}{G_{\infty}(u \,|\, a, W)} \right| \right]^2 \inprob 0$;
\item $\max_k E_{0} \left[ \sup_{u \in [0, t]}\left| \frac{S_{n,k}(t \,|\, a, W)}{S_{n,k}(u \,|\, a, W)} -\frac{S_{\infty}(t \,|\, a, W)}{S_{\infty}(u\,|\, a, W)} \right|  \right]^2 \inprob 0$.
\end{enumerate}
\item[\namedlabel{itm:bound}{(B2)}] There exists $\eta \in (0, \infty)$ such that, with probability tending to one, for $P_0$-almost all $w$, $\pi_{n,k}(a \, | \, w) \geq 1/\eta$, $\pi_\infty(a \, | \, w) \geq 1/\eta$, $G_{n,k}(t \,|\, a, w) \geq 1/\eta$, and $G_\infty(t \,|\, a, w) \geq 1/\eta$.
\item[\namedlabel{itm:dr}{(B3)}] For $P_0$-almost all $w$, there exist measurable sets $\s{S}_w, \s{G}_w \subseteq [0, t]$ such that $\s{S}_w \cup \s{G}_w = [0, t]$ and  $\Lambda_0(u \,|\, a, w) = \Lambda_{\infty}(u \,|\, a, w)$ for all $u \in \s{S}_w$ and $G_0(u \,|\, a, w) = G_{\infty}(u \,|\, a, w)$ for all $u \in \s{G}_w$. In addition, if $\s{S}_w$ is a strict subset of $[0,t]$, then $\pi_0(a \, | \, w) = \pi_{\infty}(a \, | \, w)$ as well. 
\item[\namedlabel{itm:unif_conv}{(B4)}] It holds that
\[ \max_k E_{0} \left[ \sup_{u \in [0,t]} \sup_{v\in [0,u]}\left|\frac{S_{n,k}(u \,|\, a, W)}{S_{n,k}(v \,|\, a, W)} -\frac{S_{\infty}(u \,|\, a, W)}{S_{\infty}(v\,|\, a, W)} \right|  \right]^2 \inprob 0.\]
\end{description}

\begin{thm}[Consistency]
\label{thm:consistency}
If conditions (B1)--(B3) hold, then $\theta_n(t, a) \inprob \theta_0(t, a)$. If condition (B4) also holds, then  $\sup_{u \in [0,t]} | \theta_n(u, a) -\theta_0(u, a)| \inprob 0$.
\end{thm}

Condition~\ref{itm:conv} stipulates that the estimated functions must converge in an appropriate sense to fixed limit functions, a requirement  used to control certain empirical process terms. Condition~\ref{itm:unif_conv} requires a slightly stronger condition on the convergence of $S_{n,k}$ to its limit $S_\infty$ for uniform consistency. We note that the expectations in conditions~\ref{itm:conv} and~\ref{itm:unif_conv} are with respect to $W$, and not with respect to the randomness of the nuisance estimators.  Condition~\ref{itm:bound} ensures that the estimated propensity and censoring functions are bounded uniformly away from zero in all subpopulation of patients defined by $W$. In practice, this can be guaranteed by truncating the estimated propensities and censoring probabilities. We note that there is no restriction on complexity of these nuisance function estimators; the lack of such a condition is due to the use of cross-fitting. 

Condition~\ref{itm:dr} requires that, for almost all $(t, w)$, either $S_{\infty}(t \,|\, a, w) = S_0(t \,|\, a, w)$ or \emph{both} $G_{\infty}(t \,|\, a, w) = G_0(t\,|\, a, w)$ and $\pi_\infty(a \, | \, w) = \pi_0(a \, | \, w)$. In combination with condition~\ref{itm:conv}, this implies that either $S_{n}$ or \emph{both} $G_{n}$ and $\pi_{n}$ are consistent for almost all $(t, w)$. Thus, none of the limit functions need to be identically equal to their true counterparts. This is a form of \emph{double-robustness} of the estimator $\theta_n$ to estimation of the nuisances $S_{0}$ and $(G_{0}, \pi_{0})$, because, in particular, $\theta_n$ is consistent when either $S_{n}$ is consistent everywhere or both $G_{n}$ and $\pi_{n}$ are consistent everywhere. However, condition~\ref{itm:dr} constitutes a more relaxed form of doubly-robustness akin to \emph{sequential doubly-robustness} or \emph{$2^K$-robustness} in longitudinal studies, where there are $2^K$ possible ways to achieve consistency for a $G$-computation parameter in a longitudinal study with $K$ time-points \citep{tchetgen2009drug,molina2017biometrika,luedtke2017sequential, rotnitzky2017multiply}. In our setting, there are infinitely many ways to achieve consistency. 

\subsection{Asymptotic linearity}

We now present additional conditions under which $\theta_n(t,a)$ is asymptotically linear for fixed $t$ and uniformly over $t$. We define
\begin{align*}
r_{n,t,a,1} &:= \max_k E_{0} \left| \{\pi_{n,k}(a \,|\, W) -  \pi_0(a \,|\, W) \} \{ S_{n,k}(t \,|\, a, W) - S_0(t \,|\, a, W)\}\right|;\\
r_{n,t,a,2} &:= \max_k E_0 \left| S_{n,k}(t \,|\, a, W) \int_0^t \left\{ \frac{G_0(u\,|\, a, W)}{G_{n,k}(u \,|\, a, W)}  - 1 \right\}  \left(\frac{S_0}{S_{n,k}} - 1\right)(du \,|\, a, W)  \right| .
\end{align*}
Based on these quantities, we introduce additional conditions for asymptotic linearity:
\begin{description}[style=multiline,leftmargin=2cm, labelindent=1cm]
\item[\namedlabel{itm:pointwise_rate}{(B5)}] It holds that $r_{n,t,a,1}= \fasterthan(n^{-1/2})$ and $r_{n,t,a,2} = \fasterthan(n^{-1/2})$.
\item[\namedlabel{itm:unif_rate}{(B6)}] It holds that $\sup_{u \in [0,t]} r_{n,u,a,1}  = \fasterthan(n^{-1/2})$ and $\sup_{u \in [0,t]}r_{n,u,a,2} = \fasterthan(n^{-1/2})$.
\end{description}
We have the following result concerning the asymptotic linearity of $\theta_n(t,a)$.
\begin{thm}[Asymptotic linearity]
\label{thm:asymptotic_linearity}
If conditions (B1)--(B2) hold with $S_\infty = S_0$, $G_\infty = G_0$ and $\pi_\infty = \pi_0$ and condition~\ref{itm:pointwise_rate} also holds, then $\theta_n(t,a)= \theta_0(t,a) + \d{P}_n \phi_{0,t,a}^* + \fasterthan(n^{-1/2})$. In particular, $n^{1/2}[\theta_n(t,a) - \theta_0(t,a)]$ then converges in distribution to a normal random variable with mean zero and variance $\sigma^2_{0}(t,a):=P_0\phi_{0,t,a}^{*2}$. If in addition conditions (B4) and~\ref{itm:unif_rate} also hold, then 
\[\sup_{u \in [0,t]}\left| \theta_n(u,a) - \theta_0(u,a)-  \d{P}_n \phi_{0,u,a}^*\right| = \fasterthan(n^{-1/2})\ .\]
In particular, $\left\{n^{1/2}[\theta_n(u,a) - \theta_0(u,a)] : u \in [0,t]\right\}$ then converges weakly as a process in the space $\ell^\infty([0,t])$ of uniformly bounded functions on $[0,t]$ to a tight mean zero Gaussian process with covariance function $(u,v) \mapsto P_0( \phi_{0,u,a}^*\phi_{0,v,a}^*)$.
\end{thm}
Condition~\ref{itm:pointwise_rate} requires roughly that the rates of convergence of $(S_n - S_0)(\pi_n - \pi_0)$ and $(S_n - S_0)(G_n - G_0)$ to zero be faster than $n^{-1/2}$. One approach to satisfying this condition is to assume that these nuisance functions fall in known parametric or semiparametric families such that existing estimators achieve the stipulated rates. For instance, if $S_0$ and $G_0$ follow the Cox proportional hazard model \citep{cox1972regression} and $\pi_0$ the logistic regression model, and model-based maximum likelihood estimators are used to obtain $S_n$, $G_n$ and $\pi_n$, the required rates will be achieved irrespective of the dimension of $W$. However, in practice, we recommend combining multiple candidate parametric, semiparametric and nonparametric estimators using cross-validation, as we discuss in Section~\ref{sec:nuisance}. 

\section{Pointwise and uniform inference}\label{sec:inference}

\subsection{Pointwise inference}

Theorem~\ref{thm:asymptotic_linearity} can be used to conduct asymptotically valid pointwise and uniform inference for $\theta_0(t,0)$, $\theta_0(t,1)$ and contrasts thereof. Specifically, $\theta_n^\circ(t,a) \pm z_{1-\alpha/2} \sigma_n(t,a) / \sqrt{n}$ is a Wald-type asymptotic $(1-\alpha)$-level confidence interval for $\theta_0(t,a)$, where $z_{p}$ denotes the $p$-quantile of the standard normal distribution and $\sigma_n^2(t,a) := \frac{1}{n}\sum_{k=1}^K \sum_{i \in \s{V}_{n,k}} [\phi_{n,k,t,a}(O_i) - \theta_n^\circ(t,a)]^2$ is a cross-fitted influence function-based estimator of the asymptotic variance $\sigma^2_{0}(t,a)$. However, since constructing Wald-type intervals on the logistic probability scale has been found to improve finite-sample coverage in classical settings \citep{anderson1982approximate}, we suggest this approach as well. To be precise, defining the function $\n{expit}(x) := \exp(x)/\{1 + \exp(x)\}$ for any $x \in \d{R}$ and its inverse $\n{logit}(u) := \log(u) - \log(1-u)$ for any $u \in (0,1)$, and defining $\tilde\sigma_n(t,a) := \sigma_n(t,a)/ \{\theta_n^\circ(t,a)[1-\theta_n^\circ(t,a)]\}$, we propose to use the transformed Wald-type interval 
$[\ell_n(t,a), u_n(t,a)] := \n{expit}\{ \n{logit} [\theta_n^\circ(t,a)] \pm z_{1-\alpha/2}\tilde\sigma_n(t,a) / \sqrt{n} \}$
 for any $(t,a)$ for which $\theta_n^\circ(t,a) \in (0,1)$.  If $\theta_n^\circ(t,a) = 0$, we set $[\ell_n(t,a), u_n(t,a)] := [0,\min_s\{u_n(s,a) : u_n(s,a) > 0\}]$, whereas if $\theta_n^\circ(t,a) = 1$, we set $[\ell_n(t,a), u_n(t,a)] := [\max_s\{\ell_n(s,a) : \ell_n(s,a) < 1\}, 1]$. The endpoints of this interval will be strictly contained between 0 and 1 for any $(t,a)$ such that $\theta_n^\circ(t,a) \in (0,1)$.

\subsection{Uniform inference}

If  the uniform statement of Theorem~\ref{thm:asymptotic_linearity} holds with $t = \tau$, then it can be used to construct asymptotically valid uniform confidence bands for $t\mapsto \theta_0(t, a)$ over $t \in [0,\tau]$, that is, to construct functions $t \mapsto \ell_n(t,a)$ and $t\mapsto u_n(t,a)$ such that $P_0\left\{ \ell_n(t,a) \leq \theta_0(t,a) \leq u_n(t,a)\mbox{ for all }t \in [0,\tau]\right\}$ converges to $1-\alpha$. The simplest such band is a fixed-width band with endpoints  $\theta_n^\circ(t,a) \pm c_{n,a,\alpha} / \sqrt{n}$. Here, $c_{n,a,\alpha}$ is any consistent estimator of the $(1-\alpha)$-quantile $c_{0, a,\alpha}$ of the supremum of the absolute value of the Gaussian process to which $\{ n^{1/2}[\theta_n(t,a) - \theta_0(t,a)] : t \in [0,\tau]\}$ converges weakly, that is, a mean zero Gaussian process with covariance function $(u,v)\mapsto \Sigma_0(u,v,a):=P_0( \phi_{0,u,a}^*\phi_{0,v,a}^*)$. To obtain $c_{n,a,\alpha}$, we simulate sample paths of a Gaussian process on $[0,\tau]$ with covariance function given by the cross-fitted covariance estimator 
\[(u,v) \mapsto \Sigma_n(u,v, a) := \frac{1}{n}\sum_{k=1}^K \sum_{i \in \s{V}_{n,k}} [\phi_{n,k,u,a}(O_i) - \theta_n^\circ(u,a)][\phi_{n,k,v,a}(O_i) - \theta_n^\circ(v,a)]\ .\]
We then set $c_{n,a,\alpha}$ as the sample $(1-\alpha)$-quantile of the uniform norm over $[0,\tau]$ of these sample paths. Finally, we ensure monotonicity of these bands using isotonic regression, which can only increase their coverage, as established in \cite{westling2020correcting}. While this fixed-width band is appealing in its simplicity, it does not reflect the variability in the uncertainty around $\theta_n^\circ(t,a)$ for different $t$. For instance, there is typically less uncertainty near $t = 0$, when few patients have been censored and the survival probability remains close to one, than elsewhere. An equal-width band will not reflect this.

An alternative confidence band that adapts to the variability in uncertainty over $[0, \tau]$ and is guaranteed to lie strictly within $(0,1)$ can be formed by use of standard error scaling. The proposed variable-width confidence band is given by $\n{expit}\{ \n{logit}[\theta_n^\circ(t,a)] \pm    \tilde{c}_{n,\alpha} \tilde\sigma_n(t,a) / \sqrt{n} \}.$ Here, $\tilde{c}_{n,\alpha}$ is the $(1-\alpha)$-quantile of the uniform norm over $[0,\tau]$ of the sample paths of a mean zero Gaussian process with covariance function $(u,v)\mapsto\tilde{\Sigma}_n(u,v,a):=\Sigma_n(u,v,a) / \{\tilde\sigma_n(u, a) \tilde\sigma_n(v,a)\}$. However, since $\lim_{t \to 0} \sigma_n(t,a) =\lim_{t \to t^+} \sigma_n(t,a) = 0$ for $t^+ = \inf\{t : \theta_n(t,a) =0\}$ , these sample paths are unbounded near $t = 0$ and $t = t^+$. Therefore, this method of constructing confidence bands can only produce asymptotically valid bands on intervals of the form $[t_0, t_1]$ for $t_0 > 0$ and $t_1 < t^+$. Given $[t_0, t_1]$, we then proceed in constructing the band using the approximate critical value $\tilde{c}_{n,\alpha}$ obtained as the sample $(1-\alpha)$-quantile of the uniform norms over $[t_0, t_1]$ of the above sample paths.
 As before, we ensure monotonicity of these bands using isotonic regression. In practice, we suggest choosing $t_0$ and $t_1$ based on the quantiles of the observed event times.

\subsection{Inference on causal effects}

If the pointwise statement of Theorem~\ref{thm:asymptotic_linearity} holds for both $a = 0$ and $a = 1$, then $n^{1/2}[\theta_n(t,0) - \theta_0(t,0)]$ and $n^{1/2}[\theta_n(t,1) - \theta_0(t,1)]$ converge jointly to a mean zero bivariate normal distribution. This fact can be used in conjunction with the delta method to perform inference on causal effects of the form $h(\theta_0(t,0), \theta_0(t,1))$ for any differentiable $h$. Similarly, if the uniform statement of Theorem~\ref{thm:asymptotic_linearity} holds for both $a = 0$ and $a = 1$, then the processes $\{n^{1/2}[\theta_n(t,0) - \theta_0(t,0)]:t\in[0,\tau]\}$ and $\{n^{1/2}[\theta_n(t,1) - \theta_0(t,1)]:t\in[0,\tau]\}$ converge jointly as processes to correlated Gaussian process limits, so that uniform confidence bands can be constructed for causal effects in much the same way as described above. For risk and survival ratios, these confidence bands are only valid on intervals over which the denominator is bounded away from zero.

To test the null hypothesis $H_0: \theta_0(t,0) = \theta_0(t,1)$ for all $t \in [0,\tau]$ against the complementary alternative, we propose using a test statistic of the form $n^{1/2}\int_0^\tau \left| \theta_n^\circ(t,1) - \theta_n^\circ(t,0) \right| W_n(dt)$, where $W_n$ is a user-specified, possibly data-dependent weight function. Under the null hypothesis, this test statistic converges in distribution to $\int_0^\tau \left| G_0(t)\right| W_0(dt)$ by the continuous mapping theorem for $G_0$ denoting the limiting Gaussian process of $\{n^{1/2}[\{\theta_n(t,1) - \theta_0(t,1)\} - \{\theta_n(t,0) - \theta_0(t,0)\}]  : t \in [0,\tau]\}$ and $W_0$ the deterministic in-probability limit of $W_n$. This limit distribution can be estimated by simulating Gaussian processes using the estimated covariance matrices in a similar manner as discussed above, which can then be used to find a $p$-value for the test using the observed test statistic. The user-specified weight function $W_n$ can be chosen to improve power against particular alternatives that may be expected based on the scientific context, such as early or late differences in survival, as has been done in the context of logrank tests for uninformative censoring (see, e.g.\ \citealp{harrington1982class,wu2002weighted}).

Our results can also be used to make inference on functionals of the treatment-specific survival functions. For example, a natural estimator of the treatment-specific restricted mean survival time $r_{0,a}:=\int_{0}^\tau\theta_0(t, a) \, dt$ is given  by $r_{n,a} := \int_0^\tau\theta_n^\circ(t,a)\,dt$. Uniform consistency of $\theta_n^\circ(\cdot, a)$ on $[0,\tau]$, as implied by Theorem~\ref{thm:consistency}, implies consistency of $r_{n,a}$. In view of an application of the functional delta method, the weak convergence of $\{n^{1/2}[\theta_n^\circ(t, a) - \theta_0(t, a)]:t\in[0,\tau]\}$, as implied by Theorem~\ref{thm:asymptotic_linearity},  implies that $n^{1/2} (r_{n,a} - r_{0,a})$ is asymptotically linear with influence function $o\mapsto\int_{0}^\tau \phi_{0,t,a}^*(o) \,dt$. Inference for contrasts of treatment-specific restricted mean survival times can be obtained analogously.

\section{Data-adaptive estimation of nuisance functions}\label{sec:nuisance}

As discussed above, our proposed estimator requires estimation of three nuisance parameters: the conditional survival functions $S_0$ and $G_0$ of the event time and censoring distributions, respectively, given exposure and covariates, and the propensity $\pi_0$ of exposure given covariates. We note that $\pi_0$ can be estimated using any regression estimator for a binary outcome. In practice, we recommend leveraging multiple parametric, semiparametric and nonparametric regression strategies using the SuperLearner algorithm \citep{breiman1996stacked, vanderlaan2007super}. 

There are several existing strategies for estimating $S_0$ and $G_0$. The most widely-used regression model for survival outcomes is the Cox proportional hazard model \citep{cox1972regression}, which can be used in conjunction with the Breslow estimator \citep{breslow1972estimator} or parametric estimators of the baseline cumulative hazard function to obtain estimates of $S_0$ and $G_0$. The accelerated failure time model can also be used as a semiparametric estimator of $S_0$ and $G_0$ \citep{wei1992aft}. Alternatively, various other semiparametric and nonparametric regression techniques for survival data have been proposed, including, to name a few, additive Cox models \citep{hastie1990generalized}, piecewise constant hazard models \citep{friedman1982piecewise}, and survival random forests \citep{ishwaran2008rf}. In practice, it may not be \emph{a priori} clear to the researcher which of these or other algorithms are most appropriate in a given setting. Here, we propose an iterative SuperLearner ensemble approach for combining multiple candidate nuisance estimators of $S_0$ and $G_0$.

We recall that, in view of Theorem~\ref{thm:ident},  if conditions (A1)--(A5) hold for some $a \in \{0,1\}$ and $\tau \in (0,\infty)$, then $S_0(t \, | \, a, w) = P_{0,F}(T(a) > t \, | \, W=w)$ and $G_0(t \, | \, a, w) = P_{0,F}(C(a) \geq t\, | \, W=w)$ for any $t \in [0, \tau]$. Central to our ensemble method are representations of $S_0$ and $G_0$ as minimizers of oracle loss functions, as stated in the next result. For this result, we define $\s{C}_\tau$ as the set of  functions from $[0,\tau] \times \{0,1\} \times \s{W}$ to $[0,1]$.
\begin{thm}\label{thm:loss}
Let $S^*$ be a minimizer of $S \mapsto P_0 L_{S, G_0}$ over $S \in \s{C}_\tau$ and $G^*$ be a minimizer of $G \mapsto P_0 M_{G, S_0}$ over $G \in \s{C}_\tau$, where we define the loss functions
\begin{align*}
L_{S, G}:\ &(w,a,y,\delta)\mapsto  \int_0^\tau  S(t \, | \, a,w)  \left[S(t \, | \, a,w) - 2 \left\{ 1 - \frac{\delta I(y \leq t)}{G(y \,|\, a, w)}\right\}\right]  dt \ ; \\
M_{G, S}:\ &(w, a,y, \delta)\mapsto\int_0^\tau G(t \, | \, a,w)  \left[G(t \, | \, a,w) - 2 \left\{ 1 - \frac{(1-\delta) I(y < t)}{S(y \,|\, a, w)}\right\}\right]  dt\ .
\end{align*}
If  conditions (A1)--(A5) hold for each $a \in \{0,1\}$, then $S^*(t \, | \, a,w) = S_0(t \, | \, a, w)$ for $P_0$-almost every $(a,w)$ and all $t \leq \tau$, and $G^*(t \, | \, a,w) = G_0(t \, | \, a, w)$ for $P_0$-almost every $(a,w)$ and all $t \leq \tau$ such that $S_0(t\m \, | \, a,w) > 0$.
\end{thm}

Were $G_0$ known, an optimal weighted combination of $p$ candidate estimators $S_{n}^{(1)},S_n^{(2)},\dotsc, S_n^{(p)}$ of $S_0$ could be found by minimizing the cross-validated empirical risk $\d{P}_n L_{S, G_0}$ over $S$ in the set $\Pi_S$ of convex combinations $\sum_{j=1}^p \alpha_j S_{n}^{(j)}$ for $\alpha$ in the $p$-dimensional simplex. Here, by \emph{cross-validated} we mean that the sample is split into $K$ folds, candidate estimators are each trained holding out each fold, evaluated on the held-out fold, and these held-out evaluations are used to compute the empirical mean $\d{P}_n L_{S, G}$  (see, e.g., \citealp{vanderlaan2007super} or \citealp{vanderlaan2011tmle} for additional details). Were $S_0$ known, an analogous procedure could be used to find an optimal weighted combination of $q$ candidate estimators $G_{n}^{(1)},G_n^{(2)},\dotsc, G_n^{(q)}$ of $G_0$ in the set $\Pi_G$ of convex combinations $\sum_{j=1}^q \alpha_j G_{n}^{(j)}$ for $\alpha$ in the $q$-dimensional simplex. Since $S_0$ and $G_0$ are not known in practice, we propose the following iterative strategy:
\begin{description}[style=multiline,leftmargin=2.2cm, labelindent=.5cm]
\item[\namedlabel{itm:step0}{Step 0:}] Obtain an initial estimator $G_{n,0}^*$ of $G_0$ using a nonparametric procedure.
\item[\namedlabel{itm:step1}{Step 1:}] Compute $S_{n,1}^* := \argmin_{S \in \Pi_S} \d{P}_n L_{S, G_{n,0}^*}$ and $G_{n,1}^* :=  \argmin_{G \in \Pi_G} \d{P}_n M_{G, S_{n,1}^*}.$
\item[\namedlabel{itm:stepj}{Step k:}] Compute $ S_{n,k}^* := \argmin_{S \in \Pi_S} \d{P}_n L_{S, G_{n,k-1}^*}$ and $G_{n,k}^* := \argmin_{G \in \Pi_G} \d{P}_n M_{G, S_{n,k}^*}.$
\end{description}
The procedure can be terminated, for example, when $\| S_{n,k}^* - S_{n,k-1}^*\|_{\infty}$ and $\| G_{n,k}^* - G_{n,k-1}^*\|_{\infty}$ both fall below some pre-specified threshold. In practice, we can evaluate the integrals in $L_{S,G}$ and $M_{G,S}$ using a Riemann sum over a fine grid of $t$ values.

The procedure proposed above builds on prior work, such as \cite{laan2003unified}, \cite{hothorn2005ensemble} and \cite{polley2011survsl}. However, our procedure simultaneously accomplishes several goals that, to the best of our knowledge, previous work has not. First, we do not require that the event occur on either a fully discrete or fully continuous scale, but rather allow both of these possibilities as well as mixed distributions. Second, we target the entire survival functions rather than the survival at a single point $t$. Third, we target both $S_0$ and $G_0$ rather than one or the other by iterating between optimization of $S_{n}^*$ and $G_n^*$, which has the potential to improve estimation of both.

We note that obtaining the cross-validated estimates $S_{n}^{(1)},S_n^{(2)},\dotsc, S_n^{(p)}$ and $G_{n}^{(1)},G_n^{(2)},\dotsc, G_n^{(q)}$,  a requirement for any ensemble learner, is typically the most computationally expensive step of the above procedure. In our proposed procedure, these estimates only need to be obtained once. The only computational burden of the algorithm beyond that of an ordinary SuperLearner is the possibly multiple optimization steps to find the optimal convex combinations of the candidate learners, which is typically much less computationally expensive than obtaining the cross-validated estimates of the candidate learners. Therefore,  the algorithm outlined above is not substantially more computationally expensive than an ordinary SuperLearner.

\section{Numerical studies}


We conducted a numerical study to assess the finite-sample performance of the methods proposed here. We designed our simulation procedure to mimic an observational study. We simulated a vector $W:=(W_1,W_1,W_3)$ of three continuous baseline confounders as follows. First, we simulated $W_1$, representing age, as $20 + 60 \times \n{Beta}(1.1, 1.1)$. Then, conditionally on $W_1=w_1$, we simulated independent covariates $W_2$, representing BMI, and $W_3$, representing a risk score for the event of interest, as  $18 + 32 \times \n{Beta}(1.5 + \frac{w_1}{20}, 6)$ and $10 \times \n{Beta}(1.5+ \frac{|w_1-50|}{20}, 3)$, respectively. Here, Beta$(a, b)$ represents a beta-distributed random variable with mean $a / (a + b)$. We then set $\logit P_0(A = 1 \,|\, W = w) = -1 +\log\left\{1 + \exp(-20+ \frac{w_1}{10}) + \exp(-3+\frac{w_3}{2}) \right\}$. Given $A = a$ and $W =w$, we generated the censoring variable $C$ from an exponential distribution with rate $\exp[ \beta_1+0.3 a + \log\{1 + \exp(\frac{30 -w_1}{4})\} + \frac{w_3}{ 4} ]$, where we set $\beta_1=-5.5$ to  yield an average censoring rate $E_0[P(C_0 \leq 12 \,|\, A =0, W)] =  0.2$. Here, time is considered to be measured in months, so that $t = 12$ corresponds to one year post-treatment. Given $A=0$ and $W=w$, we simulated $T$ from an exponential distribution with rate $\lambda_0(w) := \n{exp}\{\beta_0- \frac{|w_1-60|}{10}   + 2\log(w_2)+ \frac{w_3}{2}\}$.
Thus, all three covariates are predictors of risk under control. We set $\beta_0=-5.6$ to yield an average observed event rate $E_0[ P_0(T \leq C \,|\, A = 0, W)]=0.15$. Given $A = 1$ and $W = w$, we simulated the event time $T$ from a non-proportional hazards model designed to mimic the situation in which the treatment's effectiveness at preventing the event improves over a period of $r = 1.5$ months to a covariate-dependent maximal effectiveness of $\gamma(w)$, stays constant for a covariate-dependent period of time (i.e.,\ a durability) $\iota(w)$, and finally vanishes away. Parameters were chosen so that the effectiveness and durability of treatment is higher for patients with younger age and lower BMI but otherwise unrelated to the risk score. Constants in the choice of $\gamma(w)$ and $\iota(w)$ were set to produce a counterfactual risk ratio of $0.7$ at $t = 12$. The exact data-generating mechanism we used is detailed in Supplementary Material.

%

We simulated 1000 datasets using the above process for $n=500,750,\ldots,1500$. For each simulated dataset, we estimated the counterfactual survival curves using the \texttt{CFsurvival} package in \texttt{R} implementing the methods developed here. To estimate the conditional survival functions, we used the iterative SuperLearner described in Section~\ref{sec:nuisance}, implemented in the \texttt{R} package \texttt{survSuperLearner}, with a combination of parametric survival models, semiparametric proportional hazard models, and generalized additive Cox models. To estimate the propensity score, we used SuperLearner with a library consisting of generalized linear models, generalized additive models, extreme gradient boosting, and multivariate adaptive regression splines. Additional details on the candidate algorithms used for nuisance estimation are provided in Supplementary Material.

 We considered two comparator methods. First, we considered marginalizing a main-terms Cox proportional hazards model for the event time, and using the nonparametric bootstrap for inference. Second, we used the \texttt{survtmle} package \citep{benkeser2017survtmle}, which implements a state-of-the-art method for discrete-time survival analysis adjusting for baseline covariates, developed in \cite{benkeser2017improved}. For  \texttt{survtmle}, we discretized time into twelve equally-spaced intervals and used SuperLearner with generalized linear models, generalized additive models, and multivariate adaptive regression splines for nuisance estimation. For each of the three methods considered (\texttt{CFsurvival}, marginalized Cox, and \texttt{survtmle}), we recorded the estimated control and treatment survival probabilities as well as the risk ratio contrast and corresponding confidence intervals at time $t = 12$. For the proposed method, we also computed uniform confidence bands over $t\in [0,12]$.

The top row of Figure~\ref{fig:bias_var_mse} shows the bias of the three methods for each of the three parameters as a function of $n$. The bias of the proposed method was within Monte Carlo error of zero for all three parameters and all sample sizes. The bias of the marginalized Cox model estimator was relatively constant as a function of $n$, suggesting that the method is inconsistent. This was expected because the true conditional survival curves do not satisfy the proportional hazards assumption. Finally, \texttt{survtmle} demonstrated a relatively large finite-sample bias, but its bias decreased as $n$ grew. The middle and bottom rows of Figure~\ref{fig:bias_var_mse} show the variance and mean squared error (MSE) of the three methods. All variances decreased with $n$, and the variance of the proposed estimator was between that of the marginalized Cox model estimator and of \texttt{survtmle}. The MSE of the proposed method was smallest among the three methods considered for all sample sizes for both the treatment survival probability and risk ratio. For the control survival probability, our method had the smallest MSE for $n \geq 750$ and comparable MSE for $n = 500$.

\begin{figure}[h!]
\centering
\includegraphics[width = 6.5in]{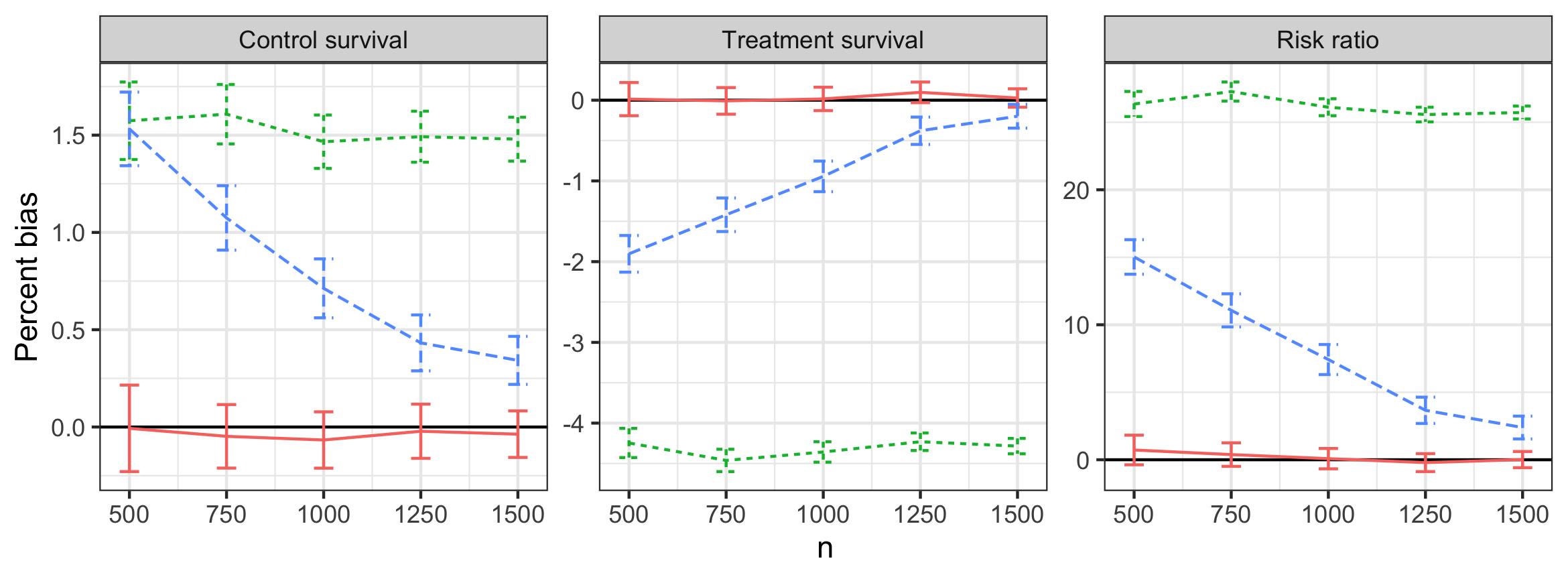}
\includegraphics[width = 6.5in]{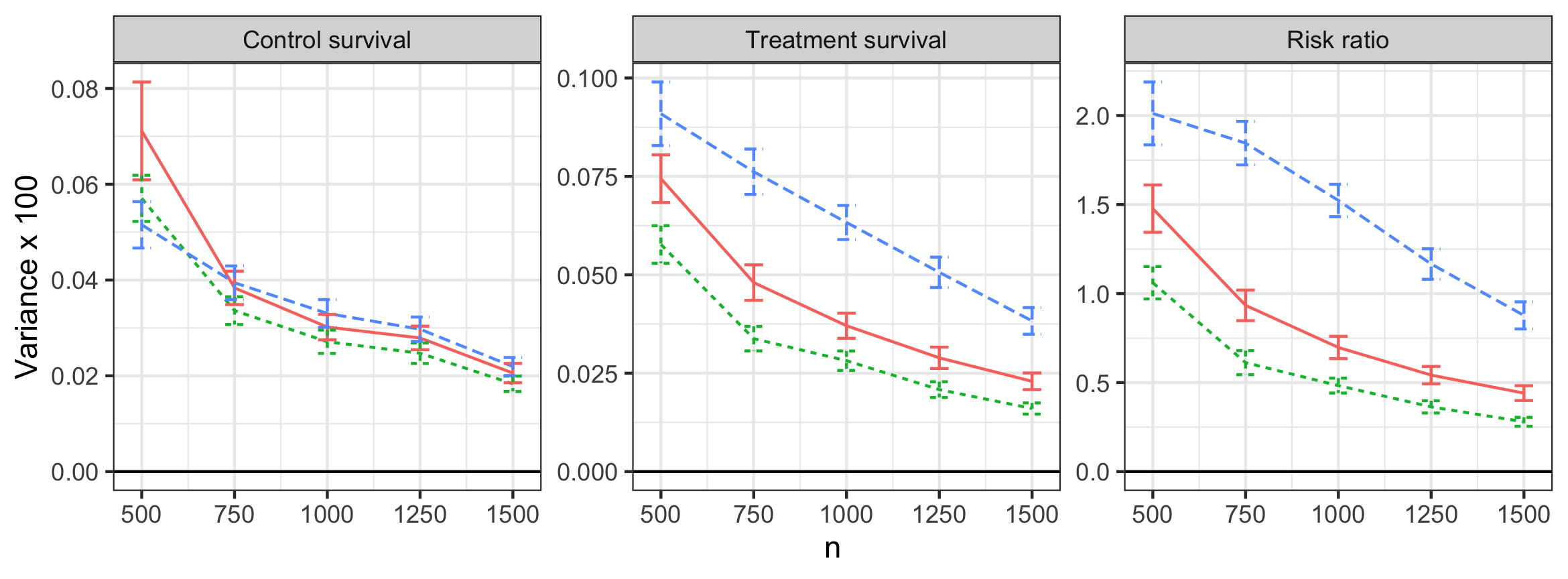}
\includegraphics[width = 6.5in]{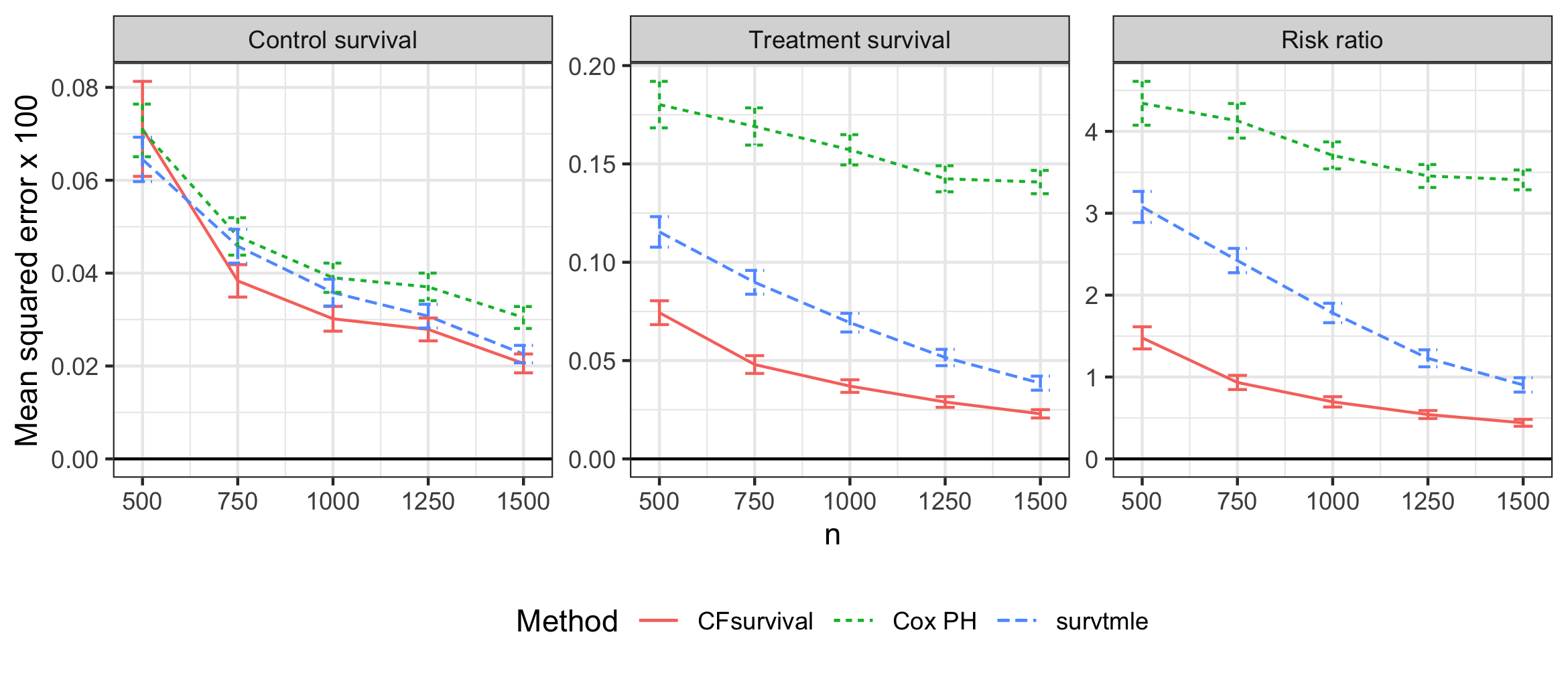}
\caption{Percent bias (top), variance (middle), and mean squared error (bottom) of the three methods considered for each of the three parameters as a function of $n$. ``CFsurvival" is the method developed here. Vertical bars represent 95\% confidence intervals taking into account Monte Carlo error induced by conducting a finite number of simulations.}
\label{fig:bias_var_mse}
\end{figure}

Figure~\ref{fig:coverage} shows the coverage of 95\% pointwise confidence intervals (CIs) at time $t = 12$ for each method and parameter considered as a function of $n$ as well as\ the coverage of 95\% uniform confidence bands over the interval $[0,12]$ based on the proposed method. The proposed method had excellent pointwise coverage for all sample sizes and parameters.  The uniform coverage for the individual survival curves was slightly anti-conservative for smaller sample size ($n\leq 750$)  but otherwise within Monte Carlo error of the nominal rate. The coverage of the bootstrap CIs for the marginalized Cox model-based  estimator was well below the nominal rate and deteriorated with growing $n$, especially for the treatment survival probability and risk ratio (for which they fall below the lower 0.7 limit of the plot for all $n$). The coverage of the \texttt{survtmle} CIs at first decreased, but eventually improved --- this may have been due to an interplay between the discretization and sample size. We emphasize that \texttt{survtmle} is designed for events occurring in discrete time, and in our simulation, the event and censoring times are both continuous. Hence, the relatively poor performance of \texttt{survtmle} reflects the insufficiency of discrete approximations for continuous events rather than inherent issues with the method. Increasing the number of grid points used in the discretization may improve the performance of \texttt{survtmle}, though it is not typically clear in practice how fine the grid should be, and there may be a bias-variance tradeoff in the grid mesh. Our method avoids these issues by allowing events to occur on an arbitrary time scale.

\begin{figure}[ht!]
\centering
\includegraphics[width = 6.5in]{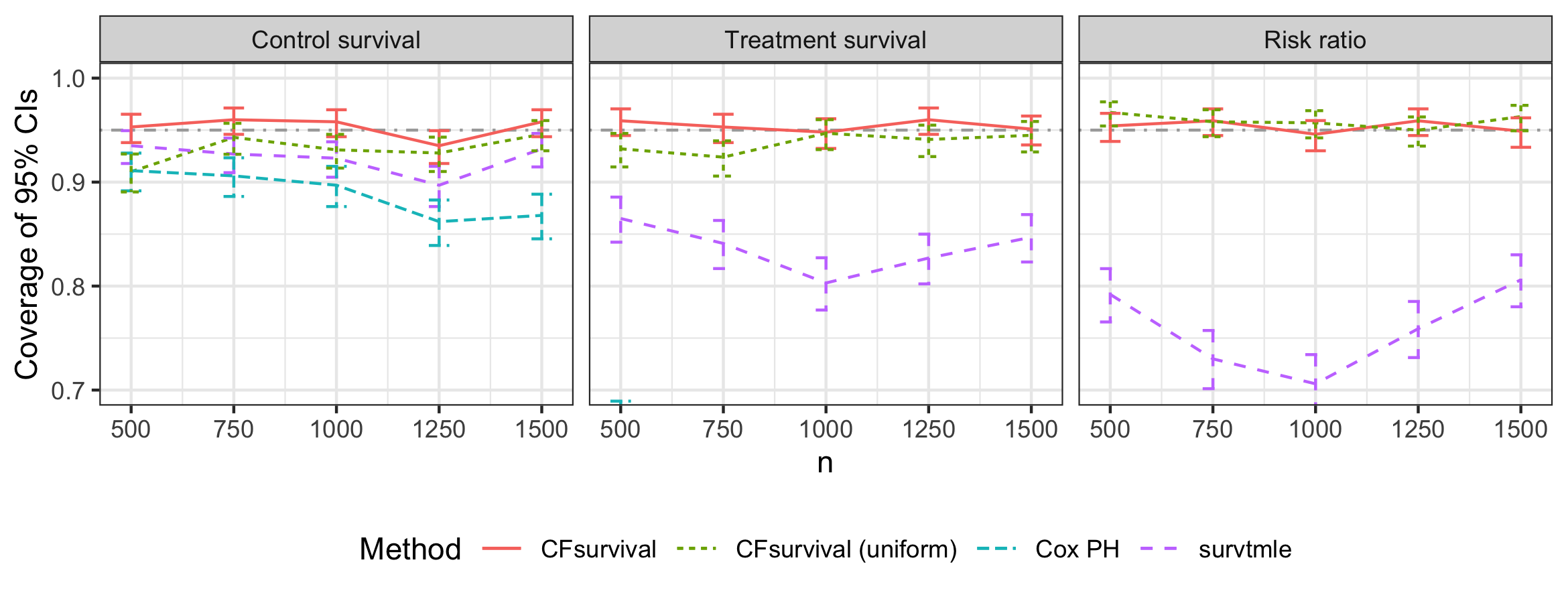}
\caption{Empirical coverage of nominal 95\% confidence intervals constructed using each of the three methods and for each of the three parameters as a function of $n$. Vertical bars represent 95\% confidence intervals taking into account Monte Carlo error induced by conducting a finite number of simulations}
\label{fig:coverage}
\end{figure}

\section{Effect of elective neck dissection on mortality}

In this section, we use the methods developed in this article to assess the effect of elective neck dissection (END) on survival among patients with clinically node-negative, high-grade parotid carcinoma. END consists of surgical removal of lymph nodes to prevent metastatic spread via the lymphatic system, and has been the subject of controversy among surgeons and oncologists. On one hand, lymph node metastases are common among patients with high-grade oral carcinomas, and END is an effective treatment for preventing these metastases. On the other hand, END is more invasive and leads to higher morbidity than radiation therapy, which can also be used to treat and prevent metastases. We refer the reader to \cite{jalisi2005management} and \cite{kowalski2007elective} for a more detailed discussion of END.

We analyzed a retrospective cohort consisting of $n = 1547$ patients in the National Cancer Database who were diagnosed with clinically node-negative, high-grade parotid cancer between January 1, 2004 and December 31, 2013, and followed until the latter date.  The exposure level $A=1$ here corresponded to receipt of END at diagnosis, and the outcome of interest was all-cause mortality up to five years post-diagnosis. Mortality was subject to right-censoring because patients could be lost to follow-up or still alive on December 31, 2013. The baseline covariate vector $W$ consisted of patient age, sex, race, tumor stage, histology, comorbidity, and payor, as well as the average income, education, county of residence, and treatment facility type. Additional details of the cohort construction and demographics may be found in \cite{harbison2020role}.

An unadjusted analysis yielded stratified Kaplan-Meier survival estimates of 56.4\% (95\% CI: 52.8--60.3) for patients receiving END and 48.6\% (43.4--54.5) for those not receiving end at $t=5$ years post-diagnosis. The survival curves were deemed to be significantly different using a log-rank test ($p < 0.0001$). These results suggest that END has a significant positive association with survival. However, since the data are observational, these results cannot be interpreted causally. By using the methods proposed here, we can adjust for baseline confounding flexibly while still reporting survival curves and contrasts thereof, which provide a simple interpretation that is familiar for many clinicians and scientists.

We used the methods presented here to estimate the treatment-specific G-computed survival functions $\theta_0(t, 0)$ and $\theta_0(t, 1)$. If the untestable causal conditions (A1)--(A5) hold, then these curves correspond to the counterfactual survival functions under assignment of all patients in the target population to no END and END, respectively. In particular, (A1)--(A5) require that the covariate vector $W$ be sufficient to control for confounding between receipt of END and mortality, and that $A$ and $W$ together be sufficient to control for the dependence between mortality and censoring. We also estimated the survival difference, survival ratio, and risk ratio functions.

We estimated the treatment propensity using SuperLearner \citep{vanderlaan2007super} with a library consisting of generalized linear models, generalized additive models, multivariate adaptive regression splines, random forests, and extreme gradient boosting. We estimated the conditional survival and censoring functions using the novel SuperLearner defined in Section~\ref{sec:nuisance} with a library consisting of the treatment group-specific Kaplan-Meier estimators, parametric survival models, Cox proportional hazard models, generalized additive models, and piecewise constant hazard models. Additional details on the libraries used for nuisance estimation and the estimated SuperLearner coefficients may be found in the Supplementary Material.

We note that the same scientific question addressed here was studied in \cite{harbison2020role} using a preliminary version of the methods developed here. However, the estimator used for the analysis presented in \cite{harbison2020role} did not use cross-fitting, and only used random forests to estimate the conditional survival and censoring functions. In addition, in \cite{harbison2020role}, uniform confidence bands or contrasts of the survival functions, which are both important for comparing the survival functions uniformly in time, were not provided.

\begin{figure}[h!]
\centering
\includegraphics[width = 6.5in]{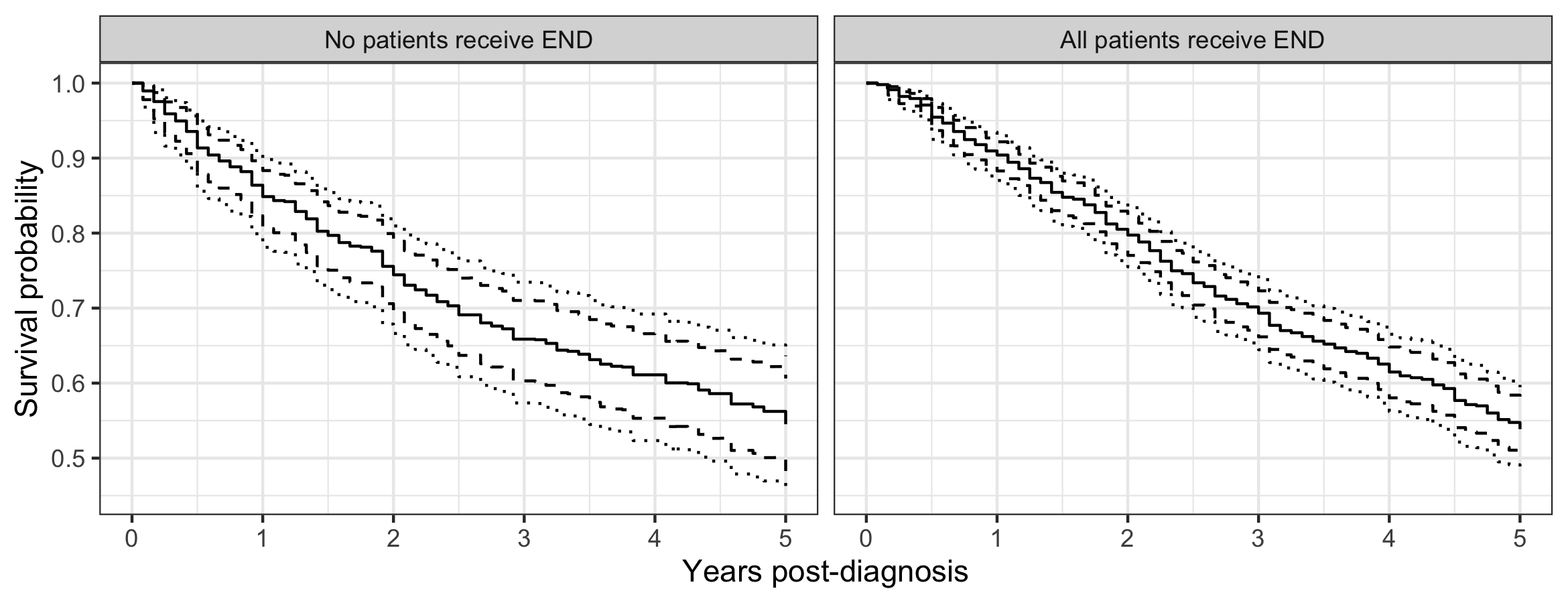}
\includegraphics[width = 3.22in]{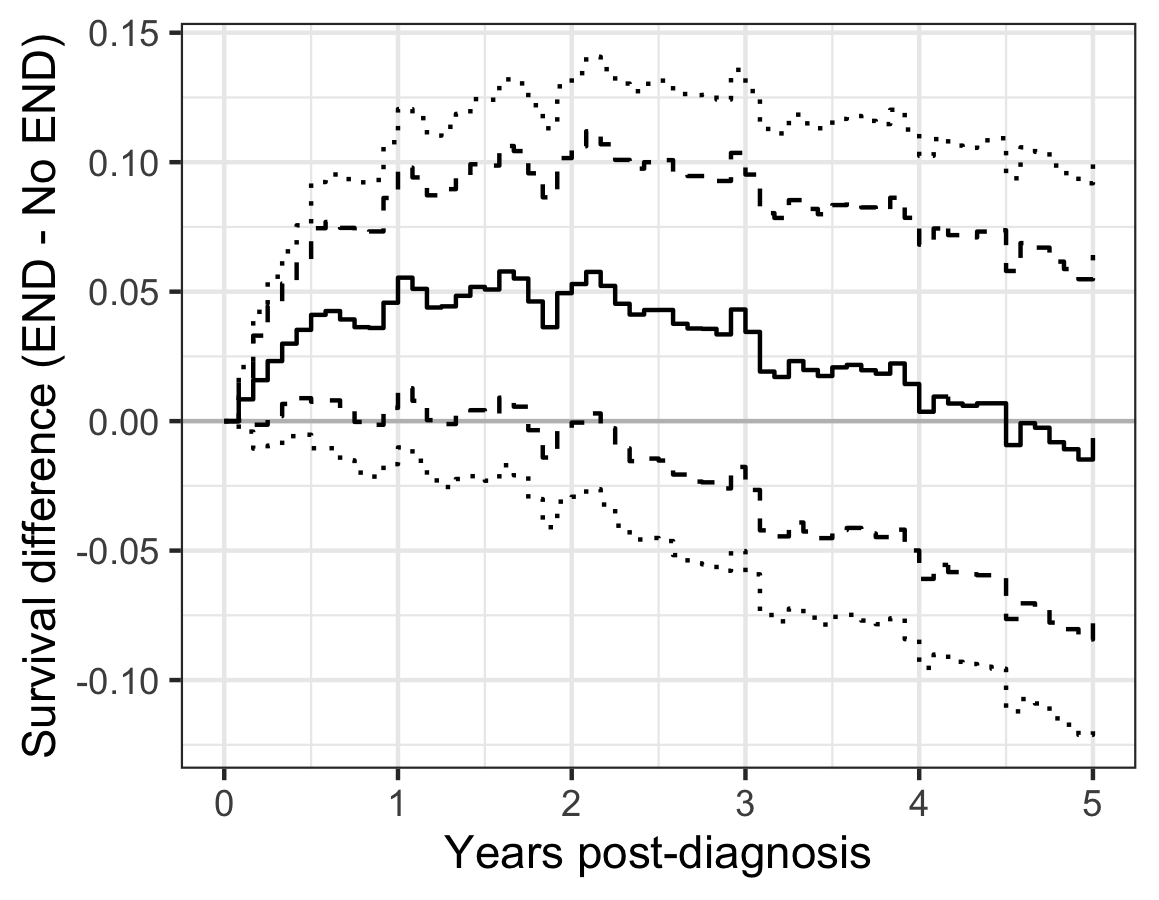}
\includegraphics[width = 3.22in]{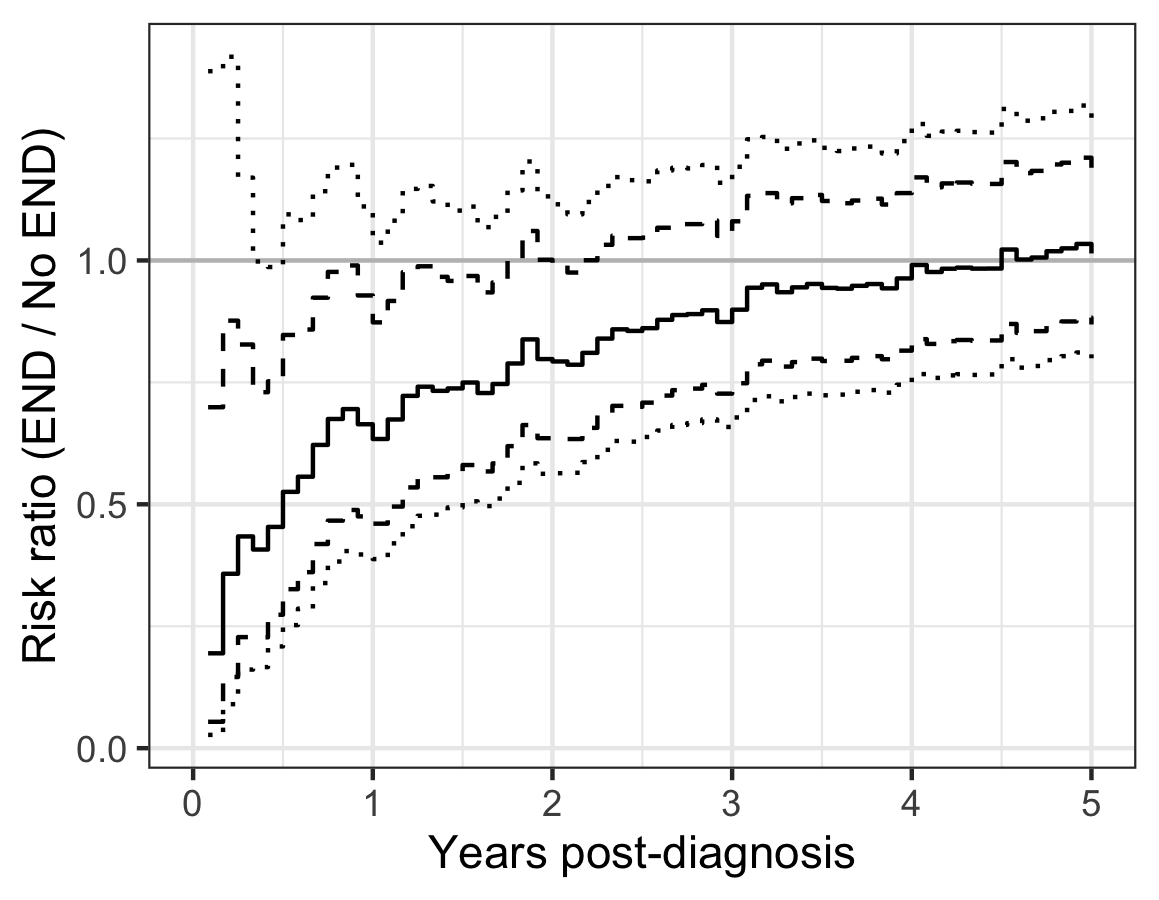}
\caption{Results of the analysis of the effect of elective neck dissection (END) on all-cause mortality. The top row shows the estimated treatment-specific survival curves were all (right panel) and no (left panel) patients to receive END. The bottom row shows the estimated survival difference (left panel) and risk ratio (right panel) functions. In all figures, pointwise 95\% confidence intervals are shown as dashed lines, and uniform 95\% confidence bands are shown as dotted lines.}
\label{fig:parotid}
\end{figure}

Figure~\ref{fig:parotid} displays the results of the analysis. The top row displays the estimated counterfactual survival functions corresponding to receiving END (left) versus not receiving END (right) along with pointwise and uniform confidence regions. We estimate that 53.9\% (95\% CI: 50.1--57.5) of patients would be alive five years post-diagnosis if undergoing END, while 54.5\% (48.3--60.6) would be alive if not undergoing END. The bottom row displays the estimated survival difference and risk ratio functions. The estimated survival difference was positive, and the estimated risk ratio was less than 1 between 0 and 4 years post-diagnosis, suggesting that END possibly improves short-term survival. However, both confidence bands included the null effect throughout this time period, and the p-value of the test of the null hypothesis that $\theta_0(t, 0) = \theta_0(t,1)$ for all $t \in [0,5]$ is 0.12. Thus, we cannot reject the null hypothesis that END does not impact overall survival through five years. The estimated survival ratio function was very similar in form to the estimated survival difference function. We estimate the restricted mean survival time through five years  to be 3.62 years (95\% CI: 3.42--3.81) under no END and  3.76 years (95\% CI: 3.65--3.87) under END, with an estimated difference of 0.14 years (95\%CI: -0.07--0.36). Therefore, after adjusting for baseline confounding, the data no longer provide evidence of an effect of END on survival.

\section{Concluding remarks}\label{conclusion}

In this article, we proposed a doubly-robust estimator of the treatment-specific survival curve in the presence of baseline confounders that permits the use of data-adaptive estimators of nuisance functions. In addition, we proposed an ensemble learner for combining multiple candidate estimators of the conditional event and censoring survival functions. We provided general sufficient conditions for consistency and asymptotic linearity, both pointwise and uniformly, of the proposed estimator, and used these results to construct confidence intervals, confidence bands and tests. The proposed methods permit event and censoring distributions that may be continuous, discrete, or mixed continuous-discrete. This is important because in many applications the event and/or censoring distributions may have either a continuous or mixed continuous-discrete support, whereas most existing methods for counterfactual survival estimation are tailored either to the fully continuous or fully discrete setting.

The methods discussed here can also be used for analyzing data from randomized trials with time-to-event outcomes. In such settings, in view of randomization, the treatment-outcome and treatment-censoring relationships are unconfounded, but the outcome-censoring relationship may still be confounded. Adjusting for baseline covariates can reduce bias due to such dependent censoring, and our methods provide a way to do so without assuming any particular form for the conditional survival and censoring functions. Comparison of our methods to other standard approaches in the context of randomized trials is an interesting and important topic of future research.

Many of the results presented here extend in a straightforward manner to the situation in which the exposure of interest is time-varying but the covariates remain fixed at baseline. However, when the exposure varies over time rather than being fixed, it is typically necessary to adjust for time-varying confounders in order to recover causal parameters, since the change in exposure status may be related to underlying changes in patients characteristics, such as health status, that are also related to the outcome or censoring mechanism. We are unaware of an extension of the continuous-time identification result we used in Theorem~\ref{thm:ident} to the setting with time-varying confounders. Instead, in the context of discrete-time longitudinal data, the nested G-formula provides an identification of the counterfactual survival probabilities \citep{robins1986}. It is unclear how or whether the methods proposed herein would extend to estimation of treatment-specific survival curves in continuous time with time-varying confounders. This is a topic of ongoing research.

\singlespacing
\section*{Acknowledgements}
The authors  gratefully acknowledge support from the University of Massachusetts Department of Mathematics and Statistics startup fund (TW) and NHLBI grant HL137808 (MC).

\singlespacing
\bibliographystyle{apa}
\typeout{}
\bibliography{mybib}

\clearpage

\begin{adjustwidth}{-.25in}{-.25in}

\section*{Supplementary Material}\vspace{.1in}

\subsection*{Additional details regarding numerical studies}\vspace{.1in}

Given $A = 1$ and $W = w$, we simulated $T$ from the survival function $t\mapsto S_0(t \,|\, 1, w) := S_0(\phi(t, w) \,|\, 0, w)$, where $t\mapsto S_0(t \,|\, 0, w)$ is the conditional survival of $T$ given $A = 0$ and $W = w$ defined in the main text, $\phi(t, w)$ is defined piecewise as
\[\begin{cases} \ t - \frac{t^2}{2r}[1-\gamma(w)]&:\ 0 \leq  t \leq r \\
                                  \  \frac{r}{2}[1 + \gamma(w)] + (t - r) \gamma(w)&:\ r \leq  t  \leq r + \iota(w)\\
                                   \ t + \frac{r}{2} [1 + \gamma(w)] + \iota(w) \gamma(w)-  [r+ \iota(w)] [2-\gamma(w)] + \frac{1}{t}[1-\gamma(w)][r+\iota(w)]^2&:\ r + \iota(w) \leq t\ , \end{cases}\]                                  
and we have set \begin{align*}
\gamma(w ) &:= \n{expit}\left[\beta_1 + \frac{1}{2}\log\left\{1 + \exp\left(\frac{w_1 - 55}{5}\right)\right\} +\frac{1}{4}\log\left\{1 + \exp\left(\frac{w_2 - 30}{3}\right)\right\}\right]\\
\iota(w) &:= \exp\left[2 - \frac{1}{2}\log\left\{1 + \exp\left(\frac{w_1 - 55}{5}\right)\right\} -\frac{1}{10}\log\left\{1 + \exp\left(\frac{w_2 - 30}{3}\right)\right\}\right].\\
\end{align*}

Table~\ref{tab:sim_surv_library} displays the candidate learners used in the SuperLearner library for estimating the conditional survival functions in the numerical studies. Table~\ref{tab:sim_prop_library} displays the candidate learners used in the SuperLearner library for estimating the propensity score in the numerical studies. 

\begin{table}[H]
\centering
\begin{tabular}{p{0.35\linewidth} | p{0.6\linewidth}}
Algorithm name & Algorithm description \\
\hline
\texttt{survSL.km} & Kaplan-Meier estimator \\
\texttt{survSL.expreg} & Survival regression assuming the event and censoring times follow exponential distributions conditional on covariates\\
\texttt{survSL.expreg.int} & Same as \texttt{survSL.expreg}, but also including interactions between treatment and each of the covariates\\
\texttt{survSL.weibreg} & Survival regression assuming the event and censoring times follow Weibull distributions conditional on covariates\\
\texttt{survSL.loglogreg} & Survival regression assuming the event and censoring times follow log-logistic distributions conditional on covariates\\
\texttt{survSL.coxph} & Main-terms Cox proportional hazards estimator with Nelson-Aalen estimator of the baseline cumulative hazard\\
\texttt{survSL.coxph.int} & Same as \texttt{survSL.coxph}, but also including interactions between treatment and each of the covariates\\
\texttt{survSL.gam} & Main-terms generalized additive Cox proportional hazards estimator as implemented in the \texttt{mgcv} package \\
\end{tabular}
\caption{Algorithms used for estimation of the conditional survival functions of event and censoring in the numerical studies.}
\label{tab:sim_surv_library}
\end{table}

\begin{table}[H]
\centering
\begin{tabular}{l | l}
Algorithm name & Algorithm description \\
\hline
\texttt{SL.mean} & Marginal mean \\
\texttt{SL.glm} & Main-terms logistic regression\\
\texttt{SL.gam} & Main-terms generalized additive model\\
\texttt{SL.earth} & Multivariate adaptive regression splines\\
\texttt{SL.xgboost} & Extreme gradient boosting
\end{tabular}
\caption{Algorithms used for estimation of the propensity score in the numerical studies.}
\label{tab:sim_prop_library}
\end{table}\vspace{.1in}

\subsection*{Additional details regarding application}\vspace{.1in}

Table~\ref{tab:app_prop_library} displays the candidate learners used in the SuperLearner library for estimating the propensity score in the analysis of the effect of elective neck dissection (END) on mortality. We used two custom screening algorithms for inputting variables into all algorithms: marginal screening (i.e.\ estimation of marginal logistic regressions of the exposure on each potential confounder) with p-value cutoffs of 0.05 and 0.10. For \texttt{SL.glm} and \texttt{SL.step}, we also input all variables.

Table~\ref{tab:app_surv_library} displays the candidate learners used in the SuperLearner library for estimating the conditional survival functions in the analysis of the effect of END on mortality.  For the \texttt{survSL.km}, \texttt{survSL.pchSL}, \texttt{survSL.coxph}, \texttt{survSL.expreg}, \texttt{survSL.weibreg}, and \texttt{survSL.loglogreg} algorithms, we included all covariates. For the \texttt{survSL.coxph}, \texttt{survSL.expreg}, \texttt{survSL.weibreg}, \texttt{survSL.loglogreg}, \texttt{survSL.gam}, and \texttt{survSL.rfsrc} algorithms, we also used two screening algorithms for inputting variables into algorithms: marginal screening, in which only variables with p-values less than 0.10 in a marginal Cox regression are included, and screening based on a penalized Cox proportional hazards model as implemented in the \texttt{glmnet} package.

For the conditional survival function of the event time, \texttt{survSL.weibreg} with all covariates received an average (across the five cross-fitting folds) of 72\% of the SuperLearner weight, while \texttt{survSL.weibreg} with the covariates selected by the penalized Cox model received an average of 28\% of the weight. For the conditional survival function of the censoring time, \texttt{survSL.km} received an average of 49\% of the SuperLearner weight, \texttt{survSL.weibreg} with all covariates selected received an average of 15\% of the weight, \texttt{survSL.weibreg} with the covariates selected by the penalized Cox model received an average of 21\% of the weight, and \texttt{survSL.coxph} with all covariates received an average of 10\% of the weight.\\

\begin{table}[H]
\centering
\begin{tabular}{p{0.35\linewidth} | p{0.6\linewidth}}
Algorithm name & Algorithm description \\
\hline
\texttt{survSL.km} & Kaplan-Meier estimator \\
\texttt{survSL.expreg} & Survival regression assuming the event and censoring times follow exponential distributions conditional on covariates\\
\texttt{survSL.weibreg} & Survival regression assuming the event and censoring times follow Weibull distributions conditional on covariates\\
\texttt{survSL.loglogreg} & Survival regression assuming the event and censoring times follow log-logistic distributions conditional on covariates\\
\texttt{survSL.coxph} & Main-terms Cox proportional hazards estimator with Nelson-Aalen estimator of the baseline cumulative hazard\\
\texttt{survSL.gam} & Main-terms generalized additive Cox proportional hazards estimator as implemented in the \texttt{mgcv} package \\
\texttt{survSL.rfsrc} & Survival random forest as implemented in the \texttt{randomForestSRC} package \\
\texttt{survSL.pchSL1}, \texttt{survSL.pchSL2}, ..., \texttt{survSL.pchSL5}& Piecewise constant hazard model with $k =1, \dotsc, 5$ bins. In these models, the conditional hazard function is assumed to be piecewise constant, and the conditional hazard in each bin is estimated using a standard SuperLearner for a binary outcome using the same library as used for the propensity score (see Table~\ref{tab:app_prop_library})
\end{tabular}
\caption{Algorithms used for estimation of the conditional survival functions of event and censoring in the parotid cancer application.}
\label{tab:app_surv_library}
\end{table}

\begin{table}[H]
\centering
\begin{tabular}{l | l}
Algorithm name & Algorithm description \\
\hline
\texttt{SL.mean} & Marginal mean \\
\texttt{SL.glm} & Main-terms logistic regression\\
\texttt{SL.step} & Forward/backwards stepwise main-terms logistic regression\\
\texttt{SL.ranger} & Random forest with 500 trees\\
\texttt{SL.gam} & Main-terms generalized additive model\\
\texttt{SL.earth} & Multivariate adaptive regression splines\\
\texttt{SL.xgboost} & Extreme gradient boosting
\end{tabular}
\caption{Algorithms used for estimation of the propensity score in the parotid cancer application.}
\label{tab:app_prop_library}
\end{table}

\newpage 
\subsection*{Proof of Theorems}\vspace{.1in}

Below, to avoid possible confusion, we use $a_0$ to denote the fixed exposure level of interest, and reserve $a$ to represent a possible realization of the exposure random variable $A$. Because this convention was not as critical in the main text as in this technical supplement, it was not used thoroughly in the main text to simplify the notation there.\vspace{.1in}

\begin{proof}[\bfseries Proof of Theorem~\ref{thm:ident}]
Conditions \ref{itm:t_exchang} and \ref{itm:a_pos} imply that 
\[P_{0,F}(T(a_0) > t \,|\, W = w) =  P_{0,F}(T(a_0) > t \,|\, A = a_0, W = w) = P_{0,F}(T > t \,|\, A = a_0, W = w) \]
for all $t \in (0,\tau]$ and $P_0$-almost every $w$, since $I(T(a_0) > t)$ is a measurable function of $T(a_0) I(T(a_0) \leq \tau)$ for $t \leq \tau$. Therefore,  $\theta_{0,F}(t,a_0) = P_{0,F}(T(a_0) > t) = E_0 \left[ P_{0,F}(T > t \,|\, A = a_0, W) \right]$ by the tower property. Let $S_{0,F}(t \,|\, a_0, w) :=  P_{0,F}(T > t \,|\, A = a_0, W = w)$. Since $T$ is a positive random variable, by Theorem~11 of \cite{gill1990product} we can then write
\[ S_{0,F}(t \,|\, a_0, w) = P_{0,F}(T > t \,|\, A = a_0, W = w)  = \Prodi_{(0,t]}\left\{ 1- \Lambda_{0,F}(du \,|\, a_0, w) \right\}\]
for $\Lambda_{0,F}(t\,|\, a_0, w) = -\int_{(0,t]} \frac{S_{0,F}(du \,|\, a_0, w)}{S_{0,F}(u\m \,|\, a_0, w)}$. Now, by definition of $Y$, $T$, and $C$, we have
\begin{align*}
R_0(t \,|\, a_0, w) &= P_0(Y \geq t \,|\, A = a_0, W = w) = P_{0,F}(T \geq t, C \geq t \,|\, A = a_0, W = w) \\
&= P_{0,F}(T(a_0) \geq t, C(a_0) \geq t \,|\, A = a_0, W = w)
\end{align*}
for all $t$. By~\ref{itm:tc_exchang}, we thus have
\begin{align*}
R_0(t \,|\, a_0, w) &=  P_{0,F}(T(a_0) \geq t \,|\, A = a_0, W = w)P_{0,F}(C(a_0) \geq t \,|\, A = a_0, W = w) \\
&= S_{0,F}(t\m \,|\, a_0, w) G_{0,F}(t\m \,|\, a_0, w)
\end{align*}
for each $t \in (0, \tau]$, where $G_{0,F}(t \,|\, a_0, w) := P_{0,F}(C(a_0) > t \,|\, A = a_0, W = w)$. We also have
\begin{align*}
F_{0,1}(t \,|\, a_0, w) &= P_0(Y \leq t, \Delta = 1 \,|\, A = a_0, W = w) \\
&= P_{0,F}(\min\{T, C\} \leq t, T \leq C \,|\, A = a_0, W = w) \\
&= P_{0,F}(T \leq t, T \leq C \,|\, A = a_0, W = w) \\
&  =\textstyle \int_{u \in (0, t]} \int_{v \geq u} P_{0,F}( du, dv \,|\, A = a_0, W = w)\ ,
\end{align*}
where here $(u,v)\mapsto P_{0,F}(u, v \,|\, A = a_0, W = w) := P_{0,F}(T \leq u, C \leq v \,|\, A = a_0, W = w)$ is the joint distribution function of $T$ and $C$ given $A =a_0$ and $W = w$. By~\ref{itm:tc_exchang}, we have
\begin{align*}
P_{0,F}(u, v \,|\, A = a_0, W = w) &= P_{0,F}(T \leq u, C \leq v \,|\, A = a_0, W = w)\\
& = P_{0,F}(T(a_0) \leq u, C(a_0) \leq v \,|\, A = a_0, W = w) \\
&= P_{0,F}(T(a_0) \leq u \,|\, A =a_0, W = w) P_{0,F}(C(a_0) \leq v \,|\, A =a_0, W = w) \\
&= \left[1 - S_{0,F}(u \,|\, a_0,w)\right] \left[1 - G_{0,F}(v \,|\, a_0,w)\right]
\end{align*}
for each $u, v\in (0,\tau]$. It follows that 
\begin{align*}
F_{0,1}(t \,|\, a_0, w)  &= -\int_{u \in (0, t]}  G_{0,F}(u\m \,|\, a_0,w) S_{0,F}(du \,|\, a_0,w)
\end{align*}
for each $t \in (0,\tau]$. As a result, we have that $F_{0,1}(dt \,|\, a_0, w) = G_{0,F}(t\m \,|\, a_0,w) S_{0,F}(dt \,|\, a_0,w)$ for each $t \in (0,\tau]$. Now, we note that~\ref{itm:c_exchang} and~\ref{itm:c_pos} together imply that $G_{0,F}(t\m \,|\, a_0, w) > 0$ for $P_0$-almost every $w$ and all $t \in [0, \tau]$. Therefore,
\begin{align*}
\Lambda_{0,F}(t\,|\, a_0, w) &= -\int_{(0,t]} \frac{S_{0,F}(du \,|\, a_0, w)}{S_{0,F}(u\m\,|\, a_0, w)}\\
&=  -\int_{(0,t]} \frac{G_{0,F}(u\m \,|\, a_0, w)S_{0,F}(du \,|\, a_0, w) }{G_{0,F}(u\m \,|\, a_0, w)S_{0,F}(u\m \,|\, a_0, w)}\\
&= \int_{(0,t]} \frac{F_{0,1}(du \,|\, a_0, w)}{R_0(u \,|\, a_0, w) }
\end{align*}
for each $t \in (0, \tau]$, which completes the proof.
\end{proof}

\begin{proof}[\bfseries Proof of Theorem~\ref{thm:eif}]

Let $\{P_{\epsilon} : |\epsilon| \leq \delta\}$ be a suitably smooth and bounded Hellinger differentiable path with $P_{\epsilon = 0} = P_0$ and score function $\dot\ell_0$ at $\epsilon = 0$. For a distribution $P$ of $(W, A, Y, \Delta)$, we let $Q$ be the marginal distribution of $W$ as implied by $P$. We then have under appropriate boundedness conditions that
\begin{align*}
\left.\frac{\partial}{\partial\epsilon}\theta_{\epsilon}(t, a_0) \right|_{\epsilon = 0} &= \left.\frac{\partial}{\partial\epsilon} \int S_{\epsilon}(t \,|\, a_0, w) dQ_{\epsilon}(w) \right|_{\epsilon = 0}\\
& =\int \left.\frac{\partial}{\partial\epsilon} S_{\epsilon}(t \,|\, a_0, w) \right|_{\epsilon = 0} dQ_{0}(w)+ \int S_{0}(t \,|\, a_0, w) \dot\ell_0(w) dQ_{0}(w)\ .
\end{align*}
The second term contributes $S_{0}(t \,|\, a_0, w)$ to the efficient influence function. 

By definition, the integrand in the first term is
\[ \left.\frac{\partial}{\partial\epsilon} \Prodi_{(0,t]} \left\{1 - \Lambda_\epsilon(du \,|\, a_0, w) \right\}  \right|_{\epsilon = 0}.\]
By Theorem~8 of \cite{gill1990product}, the product integral map $H \mapsto S_H(t) := \prodi_{(0,t]} \{1 + H(du)\}$ is Hadamard differentiable relative to the supremum norm with derivative  \[\alpha \mapsto S_H(t) \int_0^t \frac{S_H(u\m)}{S_H(u)} \alpha(du)\] at $H$. Therefore, by the chain rule, we have
\begin{align*}
\left.\frac{\partial}{\partial\epsilon} \Prodi_{(0,t]} \left\{1 - \Lambda_\epsilon(du \,|\, a_0, w) \right\}  \right|_{\epsilon = 0} = -S_0(t \,|\, a_0, w) \int_0^t \frac{S_0(u\m \,|\, a_0, w)}{S_0(u \,|\, a_0, w)} \left.\frac{\partial}{\partial\epsilon} \Lambda_\epsilon(du \,|\, a_0, w)  \right|_{\epsilon = 0}.
\end{align*}
Now, because we can write
\begin{align*}
&\left. \frac{\partial}{\partial\epsilon}  \Lambda_{\epsilon}(t \,|\, a_0, w) \right|_{\epsilon = 0} = \left. \frac{\partial}{\partial\epsilon}  \int_{(0, t]} R_{\epsilon}(u \,|\, a_0, w)^{-1} \, F_{\epsilon,1}(du \,|\,  a_0, w) \right|_{\epsilon = 0}\\
 &= \int_{(0, t]}  R_{0}(u \,|\, a_0, w)^{-1} \left. \frac{\partial}{\partial\epsilon} F_{\epsilon,1 }(du \,|\,  a_0, w)  \right|_{\epsilon = 0}-\int_{(0, t]} \left. \frac{\partial}{\partial\epsilon} R_{\epsilon}(u \,|\,  a_0, w) \right|_{\epsilon = 0}R_{0}(u \,|\,  a_0, w)^{-2} \, F_{0}(du \,|\,  a_0, w)\ ,
 \end{align*}
 we have
\begin{align*}
 \left.\frac{\partial}{\partial\epsilon} \Lambda_\epsilon(du \,|\, a_0, w)  \right|_{\epsilon = 0}&= \frac{ \left. \frac{\partial}{\partial\epsilon} F_{\epsilon,1 }(du \,|\,  a_0, w)  \right|_{\epsilon = 0}}{R_{0}(u \,|\,  a_0, w)} -\frac{\left. \frac{\partial}{\partial\epsilon} R_{\epsilon}(u \,|\,  a_0, w) \right|_{\epsilon = 0}\, F_{0}(du \,|\,  a_0, w)}{R_{0}(u \,|\,  a_0, w)^{2}}\ .
 \end{align*}
In addition,
\begin{align*}
\left. \frac{\partial}{\partial\epsilon} F_{\epsilon,1 }(u \,|\,  a_0, w)  \right|_{\epsilon = 0} &= \left. \frac{\partial}{\partial\epsilon} P_{\epsilon }( Y \leq u, \Delta = 1\,|\, A = a_0, W =w)  \right|_{\epsilon = 0}\\
&= \left. \frac{\partial}{\partial\epsilon} \iint I(y \leq u, \delta = 1) P_{\epsilon}(dy, d\delta \,|\, a_0, w)  \right|_{\epsilon = 0}\\
&=  \iint I(y \leq u, \delta = 1) \dot\ell_0(y, \delta \,|\, a_0, w) P_{0}(dy, d\delta \,|\, a_0, w)\ ,
\end{align*}
so that $ \left. \frac{\partial}{\partial\epsilon} F_{\epsilon,1 }(du \,|\,  a_0, w)  \right|_{\epsilon = 0} = \int_{\delta}I(\delta = 1) \dot\ell_0(u, \delta \,|\, a_0, w) P_{0}(du, d\delta \,|\, a_0, w)$.
In a similar manner, we find
$
 \left. \frac{\partial}{\partial\epsilon} R_{\epsilon}(u \,|\,  a_0, w)  \right|_{\epsilon = 0} =\iint I(y\geq u) \dot\ell_0(y, \delta \,|\, a_0, w) \,P_0(dy, d\delta \,|\, a_0, w)$.
Therefore,
\begin{align*}
& \left.\frac{\partial}{\partial\epsilon} \iint\prodi_{(0,t]} \left\{1 - \Lambda_\epsilon(du \,|\, a_0, w) \right\}  \, dQ_0(w) \right|_{\epsilon = 0}  \\
 &=  \iiint - I(y \leq t, \delta = 1) \frac{S_0(t \,|\, a_0, w) S_0(y\m \,|\, a_0, w)}{S_0(y \,|\, a_0, w)R_0(y \,|\, a_0, w)} \dot\ell_0(y, \delta \,|\, a_0, w) P_{0}(dy, d\delta \,|\, a_0, w) \, dQ_0(w) \\
 &\quad + \iiiint I(u \leq t, u \leq y) \frac{S_0(t \,|\, a_0, w) S_0(u\m \,|\, a_0, w)}{S_0(u \,|\, a_0, w)R_0(u \,|\, a_0, w)^2} \dot\ell_0(y, \delta \,|\, a_0, w)P_0(dy, d\delta \,|\, a_0, w) F_{0}(du \,|\, a_0, w) \, dQ_0(w) \\
 &=  \iiint - I(y \leq t, \delta = 1) \frac{S_0(t \,|\, a_0, w) S_0(y\m \,|\, a_0, w)}{S_0(y \,|\, a_0, w)R_0(y \,|\, a_0, w)} \dot\ell_0(y, \delta \,|\, a_0, w) P_{0}(dy, d\delta \,|\, a_0, w) \, dQ_0(w) \\
 &\quad + \iiint S_0(t \,|\, a_0, w) \int_0^{t\wedge y} \frac{ S_0(u\m \,|\, a_0, w)}{S_0(u \,|\, a_0, w)R_0(u \,|\, a_0, w)^2} F_{0}(du \,|\, a_0, w)\dot\ell_0(y, \delta \,|\, a_0, w)P_0(dy, d\delta \,|\, a_0, w) \, dQ_0(w) \\
 &= E_0\left[S_0(t \,|\, A, W)  \frac{I(A = a_0)}{\pi_0(a_0 \,|\, W)}\left\{   H_0(t \wedge Y, A, W)-\frac{I(Y \leq t, \Delta = 1) S_0(Y\m \,|\, A, W)}{S_0(Y \,|\, A, W)R_0(Y \,|\, A, W)} \right\}\dot\ell_0(Y, \Delta \,|\, A, W) \right],
 \end{align*}
 where $H_0(u, a, w) :=  \int_0^{u} \frac{ S_0(u\m \,|\, a, w)F_{0}(du \,|\, a, w)}{S_0(u \,|\, a, w)R_0(u \,|\, a, w)^2}$. Now, we note that 
 \begin{align*}
 E_0 \left[ \frac{I(Y \leq t, \Delta = 1) S_0(Y\m \,|\, A, W)}{S_0(Y \,|\, A, W)R_0(Y \,|\, A, W)} \,\middle|\, A  = a, W = w\right] &= \int_0^t \frac{S_0(y\m \,|\, a, w)F_0(dy \,|\, a, w)}{S_0(y \,|\, a, w)R_0(y \,|\, a, w) }
 \end{align*} 
 and $ E_0 \left[H_0(t \wedge Y, A, W) \,|\, A  = a, W = w\right]$ equals
  \begin{align*}
 & \iint_{u=0}^t I(u \leq y) \frac{S_0(u\m \,|\, a, w)F_0(du \,|\, a, w)}{S_0(u \,|\, a, w)R_0(u \,|\, a, w)^2 } \, P_0(dy \,|\, a, w) \\
 &= \int_0^t P_0(Y \geq u \,|\, A = a, W=w) \frac{S_0(u\m \,|\, a, w)F_0(du \,|\, a, w)}{S_0(u \,|\, a, w)R_0(u \,|\, a, w)^2 } \, P_0(dy \,|\, a, w) \\
 &= \int_0^t \frac{S_0(u\m \,|\, a, w)F_0(du \,|\, a, w)}{S_0(u \,|\, a, w)R_0(u \,|\, a, w) }
 \end{align*} 
 since $P_0(Y \geq u \,|\, A = a, W=w) = R_0(u \,|\, a, w)$ by definition. Therefore,
 \[E_0 \left[ H_0(t \wedge Y, A, W)-\frac{I(Y \leq t, \Delta = 1) S_0(Y\m \,|\, A, W)}{S_0(Y \,|\, A, W)R_0(Y \,|\, A, W)} \,\middle|\, A , W \right] = 0\]$P_0$-almost surely.
This implies by properties of score functions and the tower property that 
\begin{align*}
& \left.\frac{\partial}{\partial\epsilon} \int\prodi_{(0,t]} \left\{1 - \Lambda_\epsilon(u \,|\, a_0, w) \right\}  \, dQ_0(w) \right|_{\epsilon = 0}  \\
&= E_0\left[S_0(t \,|\, A, W)  \frac{I(A = a_0)}{\pi_0(a_0 \,|\, W)}\left\{H_0(t \wedge Y, A, W)-\frac{I(Y \leq t, \Delta = 1) S_0(Y\m \,|\, A, W)}{S_0(Y \,|\, A, W)R_0(Y \,|\, A, W)}  \right\} \dot\ell_0(Y, \Delta, A, W) \right].
\end{align*}Combining these results, we find that the uncentered influence function is 
\begin{align*}
&o\mapsto S_0(t \,|\,  a_0, w) \left[  1 -  \frac{I(a = a_0)}{\pi_0(a_0 \,|\, w)} \left\{ \frac{I(y\leq t, \delta = 1) S_0(y\m \,|\, a_0 , w)}{S_0(y \,|\, a_0, w)R_0(y \,|\, a_0, w)} + \int_0^{t \wedge y} \frac{ S_0(u\m \,|\, a_0, w)F_{0}(du \,|\, a_0,w)}{S_0(u \,|\, a_0, w)R_0(u \,|\, a_0,w)^2} \right\} \right].
\end{align*}
By our calculation above, the mean of the term in curly brackets is zero, and so, the mean of the entire expression is $E_0\left[S_0(t \,|\, a_0, W)\right] = \theta_0(t,a_0)$. We note that $F_{0}(du \,|\, a_0,w) / R_0(u \,|\, a_0,w) = \Lambda_0(du \,|\, a_0, w)$ and that $R_0(u \,|\, a_0, w) = S_0(u\m \,|\, a_0, w) G_0(u \,|\, a_0, w)$, so that, as claimed, the above is equal to
\begin{align*}
&S_0(t \,|\,  a_0, w) \left[ 1 -  \frac{I(a = a_0)}{\pi_0(a_0 \,|\, W)} \left\{ \frac{I(y\leq t, \delta = 1) S_0(y\m \,|\, a_0 , w)}{S_0(y \,|\, a_0, w)S_0(y\m \,|\, a_0, w)G_0(y \,|\, a_0, w)} \right.\right.\\
&\hspace{2.5in} \left.\left.+ \int_0^{t \wedge y} \frac{ S_0(u\m \,|\, a_0, w)\Lambda_{0}(du \,|\, a_0,w)}{S_0(u \,|\, a_0, w)S_0(u\m \,|\, a_0,w) G_0(u \,|\, a_0, w)} \right\} \right]  \\
&= S_0(t \,|\,  a_0, w) \left[  1 -  \frac{I(a = a_0)}{\pi_0(a_0 \,|\, w)} \left\{ \frac{I(y\leq t, \delta = 1)}{S_0(y \,|\, a_0, w)G_0(y \,|\, a_0, w)}+ \int_0^{t \wedge y} \frac{ \Lambda_{0}(du \,|\, a_0,w)}{S_0(u \,|\, a_0, w) G_0(u \,|\, a_0, w)} \right\} \right] .
\end{align*}\end{proof}
We denote by $\phi_{\infty, t}^* = \phi_{\infty, t} - \theta_0(t)$ the influence function with the limits $S_{\infty}$, $G_{\infty}$, $\Lambda_{\infty}$ and $\pi_{\infty}$ substituted for the respective nuisance parameters. We also denote by $\d{P}_{n}^k$ the empirical distribution corresponding to the $k$th validation set $\{O_i : i \in \s{V}_{n,k}\}$ and $\d{G}_{n}^k := n_k^{1/2} (\d{P}_{n}^k - P_0)$ the corresponding empirical process.

Before proving Theorems~\ref{thm:consistency} and~\ref{thm:asymptotic_linearity}, we introduce several supporting lemmas. For  nuisance functions $S$, $\pi$, $G$ and $\Lambda$ the conditional cumulative hazard corresponding to $S$, we define $\phi_{S, \pi, G, t,a_0}(w, a, \delta, y)$ as 
\[S(t \,|\, a_0, w)\left[ 1 - \frac{I(a = a_0)}{ \pi(a_0 \,|\, w)}\left\{ \frac{I(y \leq t, \delta = 1)}{S( y \,|\, a_0, w) G(y \,|\, a_0,  w)} - \int_{0}^{t \wedge y}\frac{\Lambda(du \,|\, a_0, w)}{S(u \,|\, a_0, w)G(u \,|\, a_0, w)} \right\} \right].\]
Our first result provides a useful representation of $P_0 \phi_{S, G, g, t,a_0} - \theta_0(t,a_0)$.
\begin{lemma}\label{lemma:if_mean}
For any conditional survival function $S$ and corresponding cumulative hazard $\Lambda$, any conditional censoring function $G$, and any propensity function $\pi$, $P_0 \phi_{S, G, \pi, t,a_0} - \theta_0(t,a_0)$ equals
\[E_{0} \left[ S(t \,|\, a_0, W) \int_0^t \frac{S_0(u\m \,|\, a_0, W)}{S(u \,|\, a_0, W)}\left\{\frac{\pi_0(a_0 \,|\, W) G_0(u \,|\, a_0, W)}{\pi(a_0 \,|\, W)G(u \,|\, a_0, W)} - 1\right\} (\Lambda - \Lambda_0)(du \,|\, a_0, W)   \right].\]
\end{lemma}
\begin{proof}[\bfseries{Proof of Lemma~\ref{lemma:if_mean}}]
We first write 
\[\phi_{S,G, \pi, t,a_0}(y, \delta, a, w) = S(t \,|\, a_0, w)\left\{ 1 - \frac{I(a = a_0)}{\pi(a_0 \,|\, w)} H_{S, G, t, a_0}(y, \delta, w)\right\},\]
 where we define
\[ H_{S,G,t,a_0}(y, \delta, w) := \frac{I(y \leq t, \Delta = 1)}{S(y \,|\, a_0, w) G(y \,|\, a_0, w)}- \int_0^{t \wedge y} \frac{  \Lambda(du \,|\, a_0, w)} {S(u \,|\, a_0, w) G(u \,|\, a_0, w)} \ . \]
We note that $E_{0}[ H_{S,G, t,a_0}(Y, \Delta, W) \,|\, W = w]$ equals
\begin{align*}
&\int_0^t \frac{S_0(y\m \,|\, a_0, w) G_0(y \,|\, a_0, w) }{S(y \,|\, a_0, w) G(y \,|\, a_0, w)} \Lambda_0(dy \,|\, a_0, w)- \int_0^t \frac{S_0(u\m\,|\, a_0, W) G_0(u \,|\, a_0, W)}{S(u \,|\, a_0, w)G(u \,|\, a_0, w)} \Lambda(du \,|\, a_0, w) \\
&=  -\int_0^t \frac{S_0(y\m \,|\, a_0, w) G_0(y \,|\, a_0, w) }{S(y \,|\, a_0, w) G(y \,|\, a_0, w)} (\Lambda - \Lambda_0)(dy \,|\, a_0, w)\ .
\end{align*}
Therefore, $P_0 \phi_{S, G,\pi, t, a_0} - \theta_0(t, a_0)$ equals
\begin{align*}
&E_{0} \left[ S(t \,|\, a_0, W) \left\{ 1 - \frac{I(A = a_0)}{\pi(a_0 \,|\, W)} H_{S,G, t, a_0}(Y, \Delta, W) \right\}  - S_0(t \,|\, a_0, W)\right]  \\
&= E_0\left[\frac{\pi_0(a_0 \,|\, W)}{\pi(a_0 \,|\, W)} S(t \,|\, a_0, W)  \int_0^t\frac{S_0(y\m \,|\, a_0, W) G_0(y \,|\, a_0, W) }{S(y \,|\, a_0, W) G(y \,|\, a_0, W)} (\Lambda - \Lambda_0)(dy \,|\, a_0, W) \right]\\
&\quad +E_{0} \left[ S(t \,|\, a_0, W) - S_0(t\,|\, a_0, W)\right].
\end{align*}
Now, in view of the Duhamel equation (Theorem 6 of \citealp{gill1990product}) we have
\[ S(t \,|\, a_0, w) - S_0(t\,|\, a_0, w) =  -S(t \,|\, a_0, w) \int_0^t \frac{S_0(y\m \,|\, a_0, w)}{S(y \,|\, a_0, w)} (\Lambda- \Lambda_0)(dy \,|\, a_0, w) \]
for each $(t, a_0, w)$. Therefore, combining the two terms above yields that $P_0 \phi_{S,G,\pi, t, a_0} - \theta_0(t, a_0)$ equals
\begin{align*}
 E_{0} \left[ S(t \,|\, a_0, W) \int_0^t \frac{S_0(u\m \,|\, a_0, W)}{S(u \,|\, a_0, W)}\left\{\frac{\pi_0(a_0 \,|\, W) G_0(u \,|\, a_0, W)}{\pi(a_0 \,|\, W)G(u \,|\, a_0, W)} - 1\right\} (\Lambda - \Lambda_0)(du \,|\, a_0, W)   \right].
\end{align*}
\end{proof}Next, we establish a first-order expansion of the estimator that we will make use of below.
\begin{lemma}\label{lemma:expansion}
If \ref{itm:dr}  holds, then $P_0 \phi_{\infty, t,a_0} = \theta_0(t,a_0)$, so that $\theta_{n}(t,a_0) -\theta_0(t,a_0)$ can be expressed as\begin{align*}
\d{P}_n \phi_{\infty,t,a_0}^* + \frac{1}{K} \sum_{k=1}^K \frac{K n_k^{1/2}}{n} \d{G}_{n}^k \left(\phi_{n,k,t,a_0} - \phi_{\infty,t,a_0}\right)+ \frac{1}{K} \sum_{k=1}^K \frac{K n_k}{n} \left[P_0 \phi_{n,k,t,a_0} - \theta_0(t,a_0)\right].
\end{align*}
\end{lemma}
\begin{proof}[\bfseries{Proof of Lemma~\ref{lemma:expansion}}]
By Lemma~\ref{lemma:if_mean}, $P_0 \phi_{\infty, t, a_0} - \theta_0(t, a_0)$ equals
\begin{align*}
 &E_{0} \left[ S_{\infty}(t \,|\, a_0, W) \int_0^t \frac{S_0(u\m \,|\, a_0, W)}{S_\infty(u \,|\, a_0, W)}\left\{\frac{\pi_0(a_0 \,|\, W) G_0(u \,|\, a_0, W)}{\pi_{\infty}(a_0 \,|\, W)G_\infty(u \,|\, a_0, W)} - 1\right\}  (\Lambda_{\infty} - \Lambda_0)(du \,|\, a_0, W)   \right].
\end{align*}
Now, we use \ref{itm:dr} to decompose the interval $[0,t]$ as $\s{S}_w \cup \s{S}_w^c$ for each possible value $w$ of $W$. By assumption, for $u \in \s{S}_w$, $\Lambda_0(u  \,|\, a_0, w)= \Lambda_\infty(u \,|\, a_0, w)$ so that $(\Lambda_{\infty} - \Lambda_0)(du \,|\, a_0, w) = 0$, and therefore the integral over $\s{S}_w$ is zero. If $\s{S}_w^c$ is not empty, then it is contained in $\s{G}_w$ by assumption, and for $u \in \s{G}_w$, $G_0(u \,|\, a_0, w) = G_\infty(u \,|\, a_0, w)$, and in this case $\pi_0(a_0 \,|\, w) = \pi_\infty(a_0 \,|\, w)$ by assumption as well, so that
\begin{align*}
&E_{0} \left[ S_{\infty}(t \,|\, a_0, W) \int_{\s{G}_W} \frac{S_0(u\m \,|\, a_0, W)}{S_\infty(u \,|\, a_0, W)}\left\{\frac{\pi_0(a_0 \,|\, W) G_0(u \,|\, a_0, W)}{\pi_{\infty}(a_0 \,|\, W)G_\infty(u \,|\, a_0, W)} - 1\right\} (\Lambda_{\infty} - \Lambda_0)(du \,|\, a_0, W)   \right]\\
&= E_{0} \left[ S_{\infty}(t \,|\, a_0, W) \int_{\s{S}_W} \frac{S_0(u\m \,|\, a_0, W)}{S_\infty(u \,|\, a_0, W)}\left\{\frac{\pi_0(a_0 \,|\, W) G_0(u \,|\, a_0, W)}{\pi_{0}(a_0 \,|\, W)G_0(u \,|\, a_0, W)} - 1\right\}(\Lambda_{\infty} - \Lambda_0)(du \,|\, a_0, W)   \right] = 0\ .
\end{align*}
Hence, decomposing $\int_0^t$ as $\int_{\s{S}_W}+\int_{\s{S}_W^c}$, we find that $P_0 \phi_{\infty, t, a_0}  =  \theta_0(t, a_0)$ since both integrals are zero.

To establish  the second part of the claim, we observe that $\theta_n(t,a_0) - \theta_0(t,a_0)$ can be expressed as
\begin{align*}
&\frac{1}{n} \sum_{k=1}^K \sum_{i \in \s{V}_{n,k}} \phi_{n,k, t,a_0}(O_i) - \theta_0(t,a_0)= \d{P}_n \phi_{\infty, t,a_0} - \theta_0(t,a_0) + \frac{1}{n} \sum_{k=1}^K \sum_{i \in \s{V}_{n,k}} \phi_{n,k, t,a_0}(O_i) - \d{P}_n \phi_{\infty, t,a_0} \\
&= \d{P}_n \phi_{\infty, t,a_0}^* + \frac{1}{n} \sum_{k=1}^K \sum_{i \in \s{V}_{n,k}} \left[ \phi_{n,k, t,a_0}(O_i) - \phi_{\infty, t,a_0}(O_i)\right] \\
&=  \d{P}_n \phi_{\infty, t,a_0}^* + \frac{1}{K} \sum_{k=1}^K \frac{K n_k}{n}  \frac{1}{n_k}\sum_{i \in \s{V}_{n,k}} \left[ \phi_{n,k, t,a_0}(O_i) - \phi_{\infty, t,a_0}(O_i)\right] \\
&=  \d{P}_n \phi_{\infty, t,a_0}^* + \frac{1}{K} \sum_{k=1}^K \frac{K n_k}{n} \d{P}_{n}^k \left(\phi_{n,k, t,a_0} - \phi_{\infty, t,a_0}\right) \\
&=  \d{P}_n \phi_{\infty, t,a_0}^* + \frac{1}{K} \sum_{k=1}^K \frac{K n_k}{n} (\d{P}_{n}^k - P_0) \left(\phi_{n,k, t,a_0} - \phi_{\infty, t,a_0}\right)+  \frac{1}{K} \sum_{k=1}^K \frac{K n_k}{n}P_0  \left(\phi_{n,k, t,a_0} - \phi_{\infty, t,a_0}\right)\\
&=  \d{P}_n \phi_{\infty, t,a_0}^* + \frac{1}{K} \sum_{k=1}^K \frac{K n_k^{1/2}}{n} \d{G}_{n}^k \left(\phi_{n,k, t,a_0} - \phi_{\infty, t,a_0}\right)+  \frac{1}{K} \sum_{k=1}^K \frac{K n_k}{n}\left[ P_0  \phi_{n,k, t,a_0} -\theta_0(t,a_0)\right].
\end{align*}
\end{proof}

Next, we provide a bound on the $L_2(P_0)$ distance between the estimated influence function and the limiting influence function in terms of discrepancies on the nuisance parameters. This result is useful both for demonstrating negligibility of the empirical process and second-order remainder terms.
\begin{lemma}\label{lemma:influence_bound}
If \ref{itm:bound} holds, there exists a universal constant $C(\eta)$ such that, for each $n$, $k$, $t$ and $a_0$,  $\{P_0\left( \phi_{n, k,t,a_0} - \phi_{\infty, t,a_0}\right)^2 \}^{1/2} \leq C(\eta)(A_{1,n,k,t,a_0}+A_{2,n,k,t,a_0}+A_{3,n,k,t,a_0})$, where
\begin{align*}
 A^2_{1,n,k,t,a_0}&:= E_{0}\left[ \sup_{u \in [0,t]} \left| \frac{S_{n,k}(t \,|\, a_0, W)}{S_{n,k}(u \,|\, a_0, W)} - \frac{S_{\infty}(t \,|\, a_0, W)}{S_{\infty}(u \,|\, a_0, W)} \right|  \right]^2\\
A^2_{2,n,k,t,a_0}&:= E_{0}  \left[ \frac{1}{\pi_{n,k}(a_0 \,|\, W)} - \frac{1}{\pi_{\infty}(a_0 \,|\, W)}\right]^2 \\
A^2_{3,n,k,t,a_0}&:= E_{0}\left[ \sup_{u \in [0,t]} \left| \frac{1}{G_{n,k}(u \,|\, a_0, W)} - \frac{1}{G_{\infty}(u \,|\, a_0, W)}\right| \right]^2,
\end{align*}
and $[P_0\{\sup_{u\in[0,t]}\left| \phi_{n, k,u,a_0} - \phi_{\infty, u,a_0}\right|\}^2 ]^{1/2} \leq C(\eta)(A^*_{1,n,k,t,a_0}+A_{2,n,k,t,a_0}+A_{3,n,k,t,a_0})$, where
\begin{align*}
 A^{*2}_{1,n,k,t,a_0}&:= E_{0}\left[ \sup_{u \in [0,t]}\sup_{v\in[0,u]} \left| \frac{S_{n,k}(u \,|\, a_0, W)}{S_{n,k}(v \,|\, a_0, W)} - \frac{S_{\infty}(u \,|\, a_0, W)}{S_{\infty}(v \,|\, a_0, W)} \right|  \right]^2\\
\end{align*}
\end{lemma}
\begin{proof}[\bfseries{Proof of Lemma~\ref{lemma:influence_bound}}]
First, we decompose $\phi_{n, k,t,a_0} - \phi_{\infty, t,a_0} = \sum_{j=1}^6 U_{j,n,k,t,a_0}$, where we define pointwise
\begin{align*}
U_{1,n,k,t,a_0}(o) :=&\ S_{n,k}(t \,|\, a_0, w) - S_{\infty}(t \,|\, a_0, w) \\
U_{2,n,k,t,a_0}(o) :=&\  -I(a = a_0) \left\{ \frac{1}{\pi_{n,k}(a_0 \,|\, w)} - \frac{1}{\pi_{\infty}(a_0 \,|\, w)}\right\}\left\{ \frac{I(y \leq t, \Delta =1) S_\infty(t \,|\, a_0, w)}{S_\infty(y \,|\,  a_0, w) G_\infty(y\,|\,  a_0, w)} \right.\\
&\qquad\qquad\quad\left. - \int_0^{t \wedge y} \frac{S_\infty(t \,|\, a_0, w)\Lambda_{\infty}(du \,|\,  a_0,w)}{S_\infty(u \,|\, a_0, w) G_\infty(u\,|\, a_0, w)} \right\}\\
U_{3,n,k,t,a_0}(o) :=&\ - \frac{I(a = a_0, y \leq t, \delta = 1)}{\pi_{n,k}(a_0\,|\, w)G_{\infty}(y \,|\, a_0, w)} \left\{ \frac{S_{n,k}(t \,|\, a_0, w)}{S_{n,k}(y \,|\, a_0, w)} -  \frac{S_\infty(t \,|\, a_0, w)}{S_\infty(y \,|\, a_0, w)}\right\}  \\
U_{4,n,k,t,a_0}(o) :=&\ -\frac{I(a = a_0, y \leq t, \delta = 1)S_{n,k}(t \,|\, a_0,w)}{\pi_{n,k}(a_0\,|\, w)S_{n,k}(y \,|\, a_0, w)} \left\{ \frac{1}{G_{n,k}(y \,|\, a_0, w)} - \frac{1}{G_{\infty}(y \,|\, a_0, w)}\right\} \\
U_{5,n,k,t,a_0}(o) :=&\ \frac{I(a = a_0)}{\pi_{n,k}(a_0 \,|\, W)}  \int_0^{t \wedge Y} \left\{\frac{1}{G_{n,k}(u \,|\, a_0, W)} -   \frac{1}{G_{\infty}(u \,|\, a_0, W)}  \right\} \frac{S_\infty(t \,|\, a_0, w)}{S_\infty(u \,|\, a_0, w)}\Lambda_{\infty}(du \,|\, a_0, w) \\
U_{6,n,k,t,a_0}(o) :=&\ \frac{I(a = a_0)}{\pi_{n,k}(a_0\,|\, w)} \int_0^{t \wedge y} \frac{1}{G_n(u\,|\, a_0, w)}\left\{ \frac{S_{n,k}(t \,|\, a_0, w)}{S_{n,k}(u \,|\, a_0, w)}  \Lambda_{n,k}(du \,|\, a_0, w) \right. \\
&\quad\qquad\qquad \left.- \frac{S_{\infty}(t \,|\, a_0, w)}{S_{\infty}(u \,|\, a_0, w)}  \Lambda_{\infty}(du \,|\, a_0, w)\right\}.
\end{align*}
By the triangle inequality, we have  $P_0 \left( \phi_{n, k,t,a_0} - \phi_{\infty, t,a_0}\right)^2 \leq \{\sum_{j=1}^6(P_0 U_{j,n, k, t, a_0} ^2)^{1/2}\}^2$.
We bound each term $P_0 U_{j,n, k, t, a_0} ^2$ individually.
First, since $S_{n,k}(0 \,|\, a_0, w) = S_{\infty}(0 \,|\, a_0, w) = 1$ for all $(a_0,w)$, we have
\begin{align*}
P_0U_{1,n,k,t,a_0}^2 &= E_{0} \left| S_{n,k}(t \,|\, a_0, W) - S_{\infty}(t \,|\, a_0, W) \right |^2 = E_{0} \left| \frac{S_{n,k}(t \,|\, a_0, W)}{S_{n,k}(0 \,|\, a_0, W)} - \frac{S_{\infty}(t \,|\, a_0, W)}{S_\infty(0 \,|\, a_0, W)} \right |^2\\
&\leq E_{0}\left[ \sup_{u \in [0,t]} \left|\frac{S_{n,k}(t \,|\, a_0, W)}{S_{n,k}(u \,|\, a_0, W)} - \frac{S_{\infty}(t \,|\, a_0, W)}{S_{\infty}(u \,|\, a_0, W)} \right | \right]^2.
\end{align*}
Next, noting that $y \leq t$ implies that $S_\infty(t \,|\, a_0, w) \leq S_\infty(y \,|\, a_0, w)$, and that the backwards equation (Theorem 5 of \citealp{gill1990product}) implies that $\int_0^t \frac{S(t)}{S(u)} \Lambda(du) = 1-S(t)$ for any survival function $S$, we have
\begin{align*}
P_0U_{2,n,k,t,a_0}^2 &= E_{0}\left[ I(A = a_0) \left\{ \frac{1}{\pi_{n,k}(a_0 \,|\, W)} - \frac{1}{\pi_{\infty}(a_0 \,|\, W)}\right\}^2\left\{ \frac{I(Y \leq t, \Delta =1) S_\infty(t \,|\, a_0, W)}{S_\infty(Y \,|\,  a_0, W) G_\infty(Y\,|\,  a_0, W)} \right. \right.\\
&\quad\qquad\qquad\left.\left. - \int_0^{t \wedge Y} \frac{S_\infty(t \,|\, a_0, W)\Lambda_{\infty}(du \,|\,  a_0,W)}{S_\infty(u \,|\, a_0, W) G_\infty(u\,|\, a_0, W)}\right\}^2\right] \\
&\leq \eta^2 E_{0}\left[  \left| \frac{1}{\pi_{n,k}(a_0 \,|\, W)} - \frac{1}{\pi_{\infty}(a_0 \,|\, W)}\right| \left\{ 1 + S_\infty(t \wedge Y \,|\, a_0, W) \right\}\right]^2 \\
&\leq 4\eta^2 E_{0}\left| \frac{1}{\pi_{n,k}( a_0\,|\, W)} - \frac{1}{\pi_{\infty}(a_0 \,|\, W)}\right|^2.
\end{align*} 
For the next term, we have
\begin{align*}
P_0U_{3,n,k,t,a_0}^2 &= E_{0}\left[ \frac{I(A = a_0, Y \leq t, \Delta = 1)}{\pi_{n,k}(a_0\,|\, W)^2G_{\infty}(Y \,|\, a_0, W)^2} \left\{ \frac{S_{n,k}(t \,|\, a_0, W)}{S_{n,k}(Y \,|\, a_0, W)} -  \frac{S_\infty(t \,|\, a_0, W)}{S_\infty(Y \,|\, a_0, W)}\right\}^2  \right] \\
&\leq \eta^4 E_{0}\left[ \sup_{u \in [0, t]} \left| \frac{S_{n,k}(t \,|\, a_0, W)}{S_{n,k}(u \,|\, a_0, W)} -  \frac{S_\infty(t \,|\, a_0, W)}{S_\infty(u \,|\, a_0, W)}\right| \right]^2.
\end{align*}
Similarly,
\begin{align*}
P_0U_{4,n,k,t,a_0}^2 &= E_{0}\left[  \frac{I(A = a_0, Y \leq t, \Delta = 1)S_{n,k}(t \,|\, a_0, W)^2}{\pi_{n,k}(a_0\,|\, W)^2S_{n,k}(Y \,|\, a_0, W)^2} \left\{ \frac{1}{G_{n,k}(Y \,|\, a_0, W)} -\frac{1}{G_{\infty}(Y \,|\, a_0, W)}\right\}^2 \right] \\
&\leq \eta^2E_{0}\left[ \sup_{u \in [0,t]} \left| \frac{1}{G_{n,k}(u \,|\, a_0, W)} - \frac{1}{G_{\infty}(u \,|\, a_0, W)}\right| \right]^2.
\end{align*}
 Next,
\begin{align*}
P_0U_{5,n,k,t,a_0}^2 &= E_{0}\left[    \frac{I(A=a_0)}{\pi_{n,k}(a_0\,|\, W)}  \int_0^{t \wedge Y} \left\{\frac{1}{G_{n,k}(u \,|\, a_0, W)} - \frac{1}{G_{\infty}(u \,|\, a_0, W)}  \right\}\frac{S_\infty(t \,|\, a_0, W)}{S_\infty(u \,|\, a_0, W)}\Lambda_{\infty}(du \,|\, a_0, W) \right]^2 \\
&\leq \eta^2 E_{0}\left[ \sup_{u \in [0,t]} \left| \frac{1}{G_{n,k}(u \,|\, a_0, W)} - \frac{1}{G_{\infty}(u \,|\, a_0, W)}\right|  \int_0^{t \wedge Y} \frac{S_\infty(t \,|\, a_0, W)}{S_\infty(u \,|\, a_0, W)}\Lambda_{\infty}(du \,|\, a_0, W)\right]^2 \\
&= \eta^2 E_{0}\left[ \sup_{u \in [0,t]} \left|\frac{1}{G_{n,k}(u \,|\, a_0, W)} - \frac{1}{G_{\infty}(u \,|\, a_0, W)}\right|\left| 1- S_\infty(t \wedge Y \,|\, a_0, W)  \right| \right]^2 \\
&\leq \eta^2 E_{0}\left[ \sup_{u \in [0,t]} \left| \frac{1}{G_{n,k}(u \,|\, a_0, W)} - \frac{1}{G_{\infty}(u \,|\, a_0, W)}\right| \right]^2.
\end{align*}

For the final term, we define  $B_{n,k,t}(u \,|\, a_0, w) := S_{n,k}(t \,|\, a_0, w) / S_{n,k}(u\,|\, a_0, w)$ and $B_{\infty, t}(u \,|\, a_0, w) := S_{\infty}(t \,|\, a_0, w) / S_{\infty}(u\,|\, a_0, w)$, and we note that  $B_{n,k, t}(du \,|\, a_0, w) = S_{n,k }(t \,|\, a_0, w) \Lambda_n(du \,|\, a_0, w) / S_n(u\,|\, a_0, w)$ and $B_{\infty, t}(du \,|\, a_0, w) = S_{\infty}(t \,|\, a_0, w) \Lambda_n(du \,|\, a_0, w) / S_\infty(u\,|\, a_0, w)$ by the backwards equation (Theorem 5 of \citealp{gill1990product}). Thus,
\begin{align*}
P_0U_{6,n,k,t,a_0}^2 &= E_{0}\left[   \frac{I(A = a_0)}{\pi_{n,k}(a_0 \,|\, W)^2}  \left[\int_0^{t \wedge Y} \frac{1}{G_{n,k}(u\,|\, a_0, W)}\left\{ \frac{S_{n,k}(t \,|\, a_0, W)}{S_{n,k}(u \,|\, a_0,W)}  \Lambda_{n,k}(du \,|\, a_0,W) \right.\right. \right.\\
&\quad\qquad\qquad\left.\left.\left.- \frac{S_{\infty}(t \,|\, a_0,W)}{S_{\infty}(u \,|\, a_0,W)}  \Lambda_{\infty}(du \,|\, a_0,W)\right\} \right]^2 \right] \\
&\leq \eta^2 E_{0} \left[\int_0^{t \wedge Y} \frac{1}{G_{n,k}(u\,|\, a_0, W)}\left\{ B_{n,k,t}(du \,|\, a_0,W)- B_{\infty, t}(du \,|\, a_0, W)\right\} \right]^2 .
\end{align*}
Using integration by parts, this upper bound can be re-expressed as
\begin{align*}
&\eta^2 E_{0}\left[ \frac{1}{G_n(t \wedge Y\,|\, a_0, W)}\left\{  B_{n,k,t}(t \wedge Y \,|\, a_0, W)-  B_{\infty,t}(t \wedge Y\,|\, a_0, W)\right\}  - B_{n,k,t}(0 \,|\, a_0, W)-  B_{\infty,t}(0 \,|\, a_0, W)\right.\\
&\quad\qquad\qquad\left.- \int_0^{t \wedge Y} \frac{1}{G_{n,k}(u\,|\, a_0, W)^2}\left\{ B_{n,k,t}(u \,|\, a_0, W) - B_{\infty, t}(u \,|\, a_0, W)\right\}G_{n,k}(du\,|\, a_0, W)  \right]^2 \\
&\leq  \eta^4 E_{0} \left[ \left|  \frac{S_{n,k}(t \,|\, a_0, W)}{S_{n,k}(t \wedge Y \,|\, a_0, W)}-  \frac{S_{\infty}(t \,|\, a_0, W)}{S_{\infty}(t \wedge Y\,|\, a_0, W)}\right| + \left|S_{n,k}(t \,|\, a_0, W) -  S_{\infty}(t \,|\, a_0, W)\right|\right.\\
&\quad\qquad\qquad\left.+  \sup_{u \in [0, t]} \left| \frac{S_{n,k}(t \,|\, a_0, W)}{S_{n,k}(u \,|\, a_0, W)} -  \frac{S_\infty(t \,|\, a_0, W)}{S_\infty(u \,|\, a_0, W)}\right| \frac{1}{G_{n,k}(u\,|\, a_0, W)}  \right]^2  \\
&\leq 3\eta^6 E_{0}\left[ \sup_{u \in [0, t]} \left|\frac{S_{n,k}(t \,|\, a_0, W)}{S_{n,k}(u \,|\, a_0, W)} -  \frac{S_\infty(t \,|\, a_0, W)}{S_\infty(u \,|\, a_0, W)}\right| \right]^2.
\end{align*}
The pointwise result follows. The uniform result follows from analogous calculations.
\end{proof}

Through the next several results, we demonstrate that the empirical process term $\d{G}_{n}^k \left(\phi_{n,k,t,a_0} - \phi_{\infty,t,a_0}\right)$ is $\fasterthan(n^{-1/2})$, which we use for all results that are uniform over $t \in [0, \tau]$. For this, we will need the notion of \emph{covering} and   \emph{bracketing numbers}. Given a class of functions $\s{F}$ on a sample space $\s{X}$, a norm $\| \cdot \|$,  and  $\varepsilon > 0$, the covering number $N(\varepsilon, \s{F}, \|\cdot\|)$ is the minimal number of $ \|\cdot\|$-balls of radius $\varepsilon$ needed to cover $\s{F}$. The centers of these balls need not be in $\s{F}$. An $(\varepsilon, \|\cdot\|)$ bracket is a set of the form $\{f \in \s{F}  : \ell(x) \leq f(x) \leq u(x)$ for all $x \in \s{X}\}$ such that $\|u - \ell\| \leq \varepsilon$ and is denoted $[\ell, u]$. Here, $\ell$ and $u$ need not be elements of $\s{F}$. The bracketing number $N_{[]}(\varepsilon, \s{F}, \|\cdot\|)$ is then defined as the minimal number of $(\varepsilon, \|\cdot\|)$ brackets needed to cover $\s{F}$. It is well known that $N(\varepsilon, \s{F}, \|\cdot\|) \leq N_{[]}(2\varepsilon, \s{F}, \|\cdot\|)$. Additional details regarding bracketing numbers and their relevance to empirical process theory may be found in Chapter~2 of \cite{van1996weak}.

\begin{lemma}\label{lemma:monotone_bracket}
Let $\s{F} := \{x \mapsto f_t(x): t \in [0,\tau]\}$ be a class of functions on a sample space $\s{X}$ such that $f_s(x) \leq f_t(x)$ for all $0 \leq s \leq t \leq \tau$ and $x \in \s{X}$, and such that an envelope $F$ for $\s{F}$ satisfies $\|F \|_{P, 2} < \infty$. Then $N_{[]}(\varepsilon\|F\|_{P,2}, \s{F}, L_2(P)) \leq 4 / \varepsilon^{2}$ for all $\varepsilon \in (0,1]$. If $\s{F}$ is uniformly bounded by a constant $M$, then $N(\varepsilon M, \s{F}, L_2(Q))\leq 1 / \varepsilon^2$ for each probability distribution $Q$ on $\s{X}$.
\end{lemma}
\begin{proof}[\bfseries{Proof of Lemma~\ref{lemma:monotone_bracket}}]
We note that this Lemma appears as Problem~2.7.3 in \cite{van1996weak}. The proof only relies on the fact that $t \in [0,\tau]$ through the fact that the bounds of the interval are finite, so  we can take $\tau =1$ without loss of generality. The finite second moment of the envelope function implies that $P f_0^2$ and $Pf_1^2$ are finite. Define $h:s\mapsto P( f_s - f_0)^2$. Then, $h$ is non-decreasing on $[0,1]$ with $h(0) = 0$ and $h(1) = P(f_1 - f_0)^2 \leq 4 \|F\|_{P,2}^2$.  Since $0 \leq u \leq v$ implies $(v - u)^2 \leq v^2 - u^2$, for all $0 \leq s \leq t \leq 1$, we have $P(f_t - f_s)^2  \leq P(f_t - f_0)^2 - P(f_s- f_0)^2 = h(t) - h(s)$.

Given $\varepsilon \in (0,1]$ and defining $\eta := \varepsilon \|F\|_{P,2}$, we will now produce a set of $(\eta, L_2(P))$-brackets $\{[\ell_k, u_k] : k=1,2, \dotsc, K\}$ covering $\s{F}$, where $K \leq 4/\varepsilon^2$. We do this in the following recursive manner. We start by defining $t_0 := 0$ and $\ell_1 := f_0$. If $h(0+) - h(0) \geq \eta^2$, then we set $t_1 := 0$, $I_1 := [t_0, t_1] = [0,0]$, and  $u_1 := f_0$. The bracket $[\ell_1, u_1] = [f_0, f_0]$ then has size $0 < \eta$. If $h(0+) - h(0) < \eta^2$, then the set $\{ t \leq 1 : h(t) -h(0) \leq \eta^2\}$ is necessarily non-empty, and we define $t_1 := \sup\{ t \leq 1 : h(t) -h(0) \leq \eta^2\}$. We then have $h(t_1-) - h(t_0) \leq \eta^2$ and $h(t_1+) - h(t_0) \geq \eta^2$. If $h(t_1) - h(t_0) \leq \eta^2$, we then define $I_1 := [t_0, t_1]$ and $u_1 := f_{t_1}$. Then the bracket $[\ell_1, u_1]$ satisfies $P(u_1 - \ell_1)^2 = P(f_{t_1} -f_0) \leq h(t_1) - h(0) \leq \eta^2$. If $h(t_1) - h(t_0) > \eta^2$, we define $I_1 := [t_0, t_1)$ and $u_1$ pointwise as $u_1 : x \mapsto \sup_{t < t_1} f_t(x)$. Then the bracket $[\ell_1, u_1]$ satisfies $P(u_1 - \ell_1)^2 \leq h(t_1-) - h(t_0) \leq \eta^2$ by the monotone convergence theorem.  Furthermore, in any of the cases above, the bracket $[\ell_1, u_1]$ covers $\{f_t : t \in I_1\}$ because $f_t \in [f_s, f_u]$ for any $0 \leq s \leq t \leq u \leq 1$ by the monotonicity of $t \mapsto f_t(x)$ for each $x$.

We now suppose that $I_{k}$ as an interval from $t_{k-1}$ to $t_k$ has been defined. If $I_k$ is right-closed and $t_k =  1$, then the process terminates. If $t_k < 1$ and $I_{k}$ is right-closed, we set $\ell_{k+1} := \inf_{t > t_{k}} f_t$ and $t_{k+1} := \sup\{ t \leq 1 : h(t) -h(t_k+) \leq \eta^2\}$, so that $h(t_{k+1}-) - h(t_k + ) \leq \eta^2$ and $h(t_{k+1}+) - h(t_k + )\geq \eta^2$. If $h(t_{k+1}) - h(t_k+) \leq \eta^2$, we define $I_{k+1} := (t_k, t_{k+1}]$ and $u_{k+1} := f_{t_{k+1}}$, so that $P(u_{k+1} - \ell_{k+1})^2  \leq h(t_{k+1}) - h(t_k+) \leq \eta^2$. If $h(t_{k+1}) - h(t_k+) > \eta^2$, we define $I_{k+1} := (t_{k}, t_{k+1})$ and $u_{k+1} :=  \sup_{t < t_1} f_t$, so that $P(u_{k+1} - \ell_{k+1})^2 \leq h(t_{k+1}-) - h(t_k+) \leq \eta^2$. 

If $I_k$ is right-open, we set $\ell_{k+1} := f_{t_k}$. If $h(t_k+) - h(t_k) \geq \eta^2$, then we set $t_{k+1} := t_k$, $I_{k+1} := [t_k, t_{k}]$, and $u_{k+1} := f_{t_k}$, so that the bracket $[\ell_{k+1}, u_{k+1}] = [f_{t_k}, f_{t_k}]$ has size 0. Otherwise, we set $t_{k+1} := \sup\{ t \leq 1 : h(t) -h(t_k) \leq \eta^2\}$, so that $h(t_{k+1}-) - h(t_k ) \leq \eta^2$ and $h(t_{k+1}+) - h(t_k )\geq \eta^2$. If $h(t_{k+1}) - h(t_k) \leq \eta^2$, we define $I_{k+1} := [t_k, t_{k+1}]$ and $u_{k+1} := f_{t_{k+1}}$, so that $P(u_{k+1} - \ell_{k+1})^2  \leq h(t_{k+1}) - h(t_k) \leq \eta^2$. If $h(t_{k+1}) - h(t_k) > \eta^2$, we define $I_{k+1} := [t_{k}, t_{k+1})$ and $u_{k+1} :=  \sup_{t < t_1} f_t$, so that $P(u_{k+1} - \ell_{k+1})^2 \leq h(t_{k+1}-) - h(t_k) \leq \eta^2$. 

In all of the above cases, we have that $[\ell_{k+1}, u_{k+1}]$ is an $(\eta, L_2(P))$-bracket that covers $\{f_t : t \in I_{k+1}\}$. We now define the sequence $\{r_k : k =1, 2, \dotsc\}$ depending on the form of each $I_k$ as follows. If $I_k = (t_{k-1}, t_{k})$, we define $r_k := h(t_{k}) - h(t_{k-1}+)$. If $I_k = (t_k, t_{k+1}]$, we define $r_k := h(t_{k}+) - h(t_{k-1}+)$. If $I_k = [t_{k-1}, t_{k})$, we define $r_k := h(t_{k}) - h(t_{k-1})$. Finally, if $I_k = [t_k, t_{k+1}]$, we define $r_k := h(t_{k}+) - h(t_{k-1})$. Using the definitions of $t_k$ and $I_k$ above, we then have $r_k \geq \eta^2$ for all $k$, so that $\sum_{j=1}^k r_j \geq k \eta^2$ for all $k \geq 1$. However, by a telescoping argument, we also have that for each $k \geq 1$,  $\sum_{j=1}^k r_j = h(t_k) - h(0)$ if $I_k$ is right-open and  $\sum_{j=1}^k r_j = h(t_k+) - h(0)$ if $I_k$ is right-closed. In either case, we have $k\eta^2 \leq \sum_{j=1}^k r_j \leq h(1) - h(0) \leq 4\eta^2 / \varepsilon^2$ by the definition of $\eta$. Therefore, $k \leq 4 / \varepsilon^2$, which implies that the process must terminate in a finite number of steps $K$, and that $K \leq 4/\varepsilon^2$. Furthermore, $t_{k-1} = t_k < 1$ implies that $t_{k} < t_{k+1}$, so  the process can only terminate if $I_k$ is right-closed with $t_k = 1$, which implies that $\cup_{k=1}^K I_k = [0,1]$. Therefore, $\{ [\ell_k, u_k] : k =1,\dotsc, K\}$ forms a set of no more than $4/\varepsilon^2$ $(\eta, L_2(P))$-brackets, and the union of these brackets covers $\cup_{k=1}^K\{f_t : t \in I_{k}\} = \{f_t : t \in [0,1]\} = \s{F}$. This completes the proof.

If $\s{F}$ is uniformly bounded by $M$, then we have that $\|F\|_{Q,2} = M$ for any $Q$. Then, the above implies that $N_{[]}(\varepsilon M, \s{F}, L_2(Q)) \leq 4 / \varepsilon^{2}$. Hence, in view of the basic relationship between bracketing and entropy numbers, we find that
$N(\varepsilon M, \s{F}, L_2(Q)) \leq N_{[]}( 2\varepsilon M , \s{F}, L_2(Q)) \leq 4 / (2\varepsilon)^2 = 1/\varepsilon^2$.
\end{proof}

In the next result, we use Lemma~\ref{lemma:monotone_bracket} to establish a polynomial bound on the uniform entropy numbers for the class of influence functions indexed by $t$.
\begin{lemma}\label{lemma:if_entropy}
Let $S$, $g$ and $G$ be fixed, where $t\mapsto S(t \,|\, a, w)$ is assumed to be non-increasing for each $(a,w)$, and where $G(t_0 \,|\, a_0, w) \geq 1/\eta$  and $\pi(a_0 \,|\, w) \geq 1/\eta$ for $P_0$-almost every $w$ and some $\eta \in (0, \infty)$. Then, the class of influence functions $\s{F}_{S, \pi, G, t_0,a_0} := \{ \phi_{S, \pi, G, t,a_0}  : t \in [0, t_0]\}$ satisfies
\[ \sup_Q N\left(\varepsilon \|F\|_{Q,2}, \s{F}_{S, \pi, G, t_0,a_0} , L_2(Q)\right) \leq  32/\varepsilon^{10}\]
for any $\varepsilon \in (0,1]$, where $F := 1+ 2\eta^2$ is an envelope of $\s{F}_{S, \pi, G, t_0,a_0}$, and the supremum is taken over all  distributions $Q$ on the sample space of the observed data.
\end{lemma}
\begin{proof}[\bfseries{Proof of Lemma~\ref{lemma:if_entropy}}]
We note that $\s{F}_{S, \pi, G, t_0,a_0}$ is uniformly bounded by $1 + 2\eta^2$ due to the upper bounds on $1/G$ and $1/\pi$. Therefore, we can take as our envelope function $F := 1 + 2 \eta^2$. 

We define the functions $f_t$ and $h_t$ pointwise as
\begin{align*}
f_t(w, a, \delta, y) &:=  \frac{I(a = a_0, y \leq t, \delta = 1)S(t\,|\, a_0, w)}{ \pi(a \,|\, w)S( y \,|\, a, w) G(y \,|\, a,  w)}; \\
h_t(w,a,y) &:= \int \frac{I(a = a_0, u \leq y, u \leq t)S(t\,|\, a_0, w)}{ \pi(a_0 \,|\, w)S( u \,|\, a_0, w) G(u \,|\, a_0,  w)} \Lambda(du \,|\, a_0, w).
\end{align*}
We then define the function classes $\s{F}_{1, S, t_0,a_0} := \{ w \mapsto S(t \,|\, a_0, w) : t \in [0,t_0] \}$, 
$\s{F}_{2, S, \pi, G, t_0,a_0} := \{ (w, a, \delta, y) \mapsto f_t(w,a,\delta,y)  : t \in [0,t_0]\}$ and
$\s{F}_{3, S, \pi, G, t_0,a_0} := \{ (w, a, y) \mapsto  h_t(w,a,y) : t \in [0,t_0]\}$. We can then write 
\[\s{F}_{S, \pi, G, t_0,a_0} \subseteq \left\{f_1 - f_2 + f_3 : f_1 \in\s{F}_{1, S, t_0,a_0} , f_2 \in   \s{F}_{2, S, \pi, G, t_0,a_0}, f_3 \in \s{F}_{3, S, \pi, G, t_0,a_0}\right\}. \]
Since $u\mapsto S(u \,|\, a, w)$ is non-increasing for all $(a,w)$ and uniformly bounded by 1, Lemma~\ref{lemma:monotone_bracket} implies that $\sup_Q N(\varepsilon, \s{F}_{1, S, t_0,a_0}, L_2(Q)) \leq 2/\varepsilon^2$. The next two classes are slightly more complicated to study.

We note that $\s{F}_{2, S, \pi, G, t_0,a_0}$ is contained in the product of the classes $\{y \mapsto I(y \leq t) : t \in [0, t_0]\}$, $\s{F}_{1, S, t_0, a_0}$,  and the singleton class $\{ (w, a, \delta, y) \mapsto I(a = a_0, \delta = 1, y \leq t_0) /[ \pi(a_0 \,|\, w)S( y \,|\, a_0, w) G(y \,|\, a_0,  w)]\}$. The first two classes both have covering numbers $\sup_Q N(\varepsilon, \cdot, L_2(Q)) \leq 2 / \varepsilon^2$ by Lemma~\ref{lemma:monotone_bracket}. The third class has uniform covering number 1 for all $\varepsilon$ because it can be covered with a single ball of any positive radius. In addition, $\s{F}_{2, S, \pi, G, t_0,a_0}$ is uniformly bounded by $\eta^2$. Therefore, Lemma~5.1 of \cite{vaart2006survival} (see also Theorem 2.10.20 in \citealp{van1996weak}) implies that $\sup_Q N(\varepsilon\eta^2,\s{F}_{2, S, \pi, G, t_0,a_0}, L_2(Q)) \leq 4 / \varepsilon^4$.

Finally, we turn to $\s{F}_{3, S, \pi, G, t_0,a_0}$. Define the distribution function $\mu^*:t\mapsto 1-E_0[S(t \,|\, a_0, W)]$ and note that the probability measure defined by the distribution function $t\mapsto 1-S(t \,|\, a_0, w)$ is dominated by $\mu^*$ for $P_0$-almost all $w$ because if $\s{U}$ is a set such that $\mu^*(\s{U}) = 0$ then
\[ 0 = \int_{\s{U}} \mu^*(du)  = -\int_{\s{U}} E_0[S(du \,|\, a_0, W)] = -E_0 \left[ \int_{\s{U}}  S(du \,|\, a_0, W) \right] \]
by Fubini's theorem, which implies that $\int_{\s{U}}  S(du \,|\, a_0, w) = 0$ for $P_0$-almost every $w$ since $\int_{\s{U}}  S(du \,|\, a_0, w) \leq 0$ for all $w$. This implies that $\Lambda( \cdot \,|\, a_0, w)$ is also dominated by $\mu^*$ for $P_0$-almost all $w$. Hence, we can write
\[ h_t(w,a,y) = \int \frac{I(a = a_0, u \leq y, u \leq t)S(t\,|\, a_0, w)}{ \pi(a_0 \,|\, w)S( u \,|\, a_0, w) G(u \,|\, a_0,  w)} \lambda^*(u \,|\, a_0, w) \mu^*(du),\]
where $\lambda^*(du \,|\, a_0, w) := \Lambda(du \,|\, a_0, w) / \mu^*(du)$ is the Radon-Nikodym derivative of $\Lambda(\cdot \,|\, a_0, w)$ with respect to $\mu^*$. Furthermore, $\mu^*$ and $\lambda^*$ are fixed since they only depend on $P_0$ and $S$, which are fixed by assumption. We can then write $\s{F}_{3, S, \pi, G, t_0,a_0} := \{ (w, a, y) \mapsto \int m_t(u, w, a, y) \mu^*(du) : t \in [0, t_0] \}$, where we have defined
\[ m_t(u, w, a, y) := \frac{I(a = a_0, u \leq y, u \leq t)S(t\,|\, a_0, w)}{ \pi(a_0 \,|\, w)S( u \,|\, a_0, w) G(u \,|\, a_0,  w)} \lambda^*(u \,|\, a_0, w).\]
The class of functions $\s{M}_{t_0} := \{m_t : t \in [0, t_0]\}$ is then contained in the product of the singleton class $\{(w, a, y, u) \mapsto  I(a = a_0, u \leq y) \lambda^*(u \,|\, a_0, w)  /[ \pi(a_0 \,|\, w)S( u \,|\, a_0, w) G(u \,|\, a_0,  w)]\}$ and the classes $\{u \mapsto I(u \leq t) : t \in [0,t_0]\}$ and $\{ w \mapsto S(t\,|\, a_0, w) : t \in [0, t_0]\}$, which as discussed above both have $L_2(Q)$ covering number bounded by $2/\varepsilon^2$ for any probability measure $Q$. Therefore, by an analogous argument to that above, $N(\varepsilon \eta^2,\s{M}_{t_0}, L_2(Q)) \leq 4 / \varepsilon^4$ for every probability measure $Q$.  We next note that by Jensen's inequality, $\|h_t - h_s\|_{L_2(Q)} \leq \|m_t - m_s\|_{L_2(\mu^* \times Q)}$ for any probability measure $Q$, which implies that $N(\varepsilon \eta^2, \s{F}_{3, S, \pi, G, t_0,a_0}, L_2(Q)) \leq N(\varepsilon \eta^2, \s{M}_t, L_2(\mu^* \times Q)) \leq 4 / \varepsilon^4$ for all $Q$.

We have shown that $\s{F}_{S, \pi, G, t_0,a_0}$ is a sum of three classes with uniform covering numbers bounded by $2/\varepsilon^2$, $4 / \varepsilon^4$ and $4 / \varepsilon^4$, respectively. Therefore, by Lemma~5.1 of \cite{vaart2006survival}, 
\[\sup_Q N(\varepsilon (1 + 2\eta^2), \s{F}_{S, \pi, G, t_0,a_0}, L_2(Q)) \leq 32/\varepsilon^{10},\]
as claimed.

\end{proof}

\begin{lemma}\label{lemma:emp_process}
If \ref{itm:bound}--\ref{itm:conv} hold, then 
$\frac{1}{K} \sum_{k=1}^K \frac{Kn_k^{1/2}}{n}  \d{G}_{n}^k \left(\phi_{n,k,t,a_0} - \phi_{\infty,t,a_0}\right)  = \fasterthan(n^{-1/2})$.
If \ref{itm:unif_conv} also holds, then
\[\frac{1}{K} \sum_{k=1}^K \frac{Kn_k^{1/2}}{n} \sup_{u \in [0,t]} \left| \d{G}_{n}^k \left(\phi_{n,k,u,a_0} - \phi_{\infty,u,a_0}\right) \right| = \fasterthan(n^{-1/2}).\]
\end{lemma}
\begin{proof}[\bfseries{Proof of Lemma~\ref{lemma:emp_process}}]
We first note that
\[ \frac{Kn_k^{1/2}}{n} \leq \frac{K\left( \left|n_k - n/K\right| +n/K \right) ^{1/2}}{n} \leq \frac{K \left|n_k - n/K\right|^{1/2} + K\left(n/K \right) ^{1/2}}{n}  \leq \frac{K}{n} + \left( \frac{K}{n}\right)^{1/2}\]
for all $k$ since $|n_k - n/K| \leq 1$  by assumption. Therefore, we have that
\[\left|\frac{1}{K} \sum_{k=1}^K \frac{Kn_k^{1/2}}{n}  \d{G}_{n}^k \left(\phi_{n,k,t,a_0} - \phi_{\infty,t,a_0}\right)  \right|\leq \boundeddet(n^{-1/2}) \frac{1}{K} \sum_{k=1}^K   \left| \d{G}_{n}^k \left(\phi_{n,k,t,a_0} - \phi_{\infty,t,a_0}\right)  \right|\]
since $K = \boundeddet(1)$. For the first claim, it suffices to show that
$\frac{1}{K} \sum_{k=1}^K   \left| \d{G}_{n}^k \left(\phi_{n,k,t,a_0} - \phi_{\infty,t,a_0}\right)  \right| =\fasterthan(1)$.
Using the law of iterated expectation, we write
\begin{align*}
 E_0  \left| \d{G}_{n}^k \left(\phi_{n,k,t,a_0} - \phi_{\infty,t,a_0}\right)  \right| &=  E_0\left[ E_0 \left[   \left| \d{G}_{n}^k \left(\phi_{n,k,t,a_0} - \phi_{\infty,t,a_0}\right)  \right| \,\middle|\, \s{T}_{n,k} \right] \right]= E_0\left[ E_0 \left[ \sup_{f \in \s{F}_{n,k,t,a}}  | \d{G}_{n}^k f | \,\middle|\, \s{T}_{n,k} \right] \right]
 \end{align*}
 with $ \s{F}_{n,k,t,a} $ denoting the class of functions containing $\phi_{n,k,t,a_0} - \phi_{\infty,t,a_0}$, which is a singleton class because $\phi_{n,k,t,a}$ is a fixed function when conditioning on the training set $\s{T}_{n,k}$. We will apply Theorem~2.14.1 of \cite{van1996weak} to bound the inner expectation. The covering number of this class is 1 for all $\epsilon$, so the uniform entropy integral $J(1,  \s{F}_{n,k,t,a_0})$ is 1 relative to the natural envelope $|\phi_{n,k,t,a_0} - \phi_{\infty,t,a_0}|$. Therefore, there is a universal constant $C'$ such that
 \begin{align*}
 E_0\left[E_0 \left[ \sup_{f \in \s{F}_{n,k,t,a}}  | \d{G}_{n}^k f  | \,\middle|\, \s{T}_{n,k} \right] \right]&\leq C' E_0\left[ E_0 \left[ \left\{\phi_{n,k,t,a_0}(O) - \phi_{\infty,t,a_0}(O)\right\}^2 \,\middle|\, \s{T}_{n,k} \right]^{1/2} \right].
\end{align*}
By Jensen's inequality and another application of the tower property, this is bounded by
\begin{align*}
 C' \left\{ E_0\left[ E_0 \left[ \left\{\phi_{n,k,t,a_0}(O) - \phi_{\infty,t,a_0}(O)\right\}^2 \,\middle|\, \s{T}_{n,k} \right] \right] \right\}^{1/2} &= C'  \left\{ E_0\left[ \left\{\phi_{n,k,t,a_0}(O) - \phi_{\infty,t,a_0}(O)\right\}^2  \right] \right\}^{1/2}.
 \end{align*}
By Lemma~\ref{lemma:influence_bound}, we therefore have that $\frac{1}{K} \sum_{k=1}^K   \left| \d{G}_{n}^k \left(\phi_{n,k,t,a_0} - \phi_{\infty,t,a_0}\right)  \right|$ is bounded above by $CC'\max_k C_k$, where $C_k:=C_{1k}+C_{2k}+C_{3k}$ with 
\begin{align*}
C^2_{1k}&:=E_{0}\left[ \sup_{u \in [0,t]} \left| \frac{S_{n,k}(t \,|\, a_0, W)}{S_{n,k}(u \,|\, a_0, W)} - \frac{S_{\infty}(t \,|\, a_0, W)}{S_{\infty}(u \,|\, a_0, W)} \right |  \right]^2\\
C^2_{2k}&:=E_{0}  \left| \frac{1}{\pi_{n,k}(a_0 \,|\, W)} - \frac{1}{\pi_{\infty}(a_0 \,|\, W)}\right|^2 \\
C^2_{3k}&:=E_{0}\left[ \sup_{u \in [0,t]} \left| \frac{1}{G_{n,k}(u \,|\, a_0, W)} - \frac{1}{G_{\infty}(u \,|\, a_0, W)}\right| \right]^2.
\end{align*}
By~\ref{itm:conv}, this upper bound tends  to zero in probability.
 
 For the uniform statement, we need to show that
 \[\frac{1}{K} \sum_{k=1}^K  \sup_{u \in [0,t]} \left| \d{G}_{n}^k \left(\phi_{n,k,u,a_0} - \phi_{\infty,u,a_0}\right)  \right| =\fasterthan(1).\]
 The basic argument is the same. Using the law of iterated expectation, we write
\begin{align*}
 E_0\left[  \sup_{u \in [0,t]} \left| \d{G}_{n}^k \left(\phi_{n,k,u,a_0} -\phi_{\infty,u,a_0}\right)  \right| \right] &=  E_0\left[ E_0 \left[  \sup_{u \in [0,t]} \left| \d{G}_{n}^k \left(\phi_{n,k,u,a_0} - \phi_{\infty,u,a_0}\right)  \right| \,\middle|\, \s{T}_{n,k} \right] \right]\\
 & = E_0\left[ E_0 \left[ \sup_{g \in \s{G}_{n,k,t,a_0}}  | \d{G}_{n}^k g  | \,\middle|\, \s{T}_{n,k} \right] \right],
 \end{align*}
 where $ \s{G}_{n,k,t,a_0} := \{\phi_{n,k,u,a_0} - \phi_{\infty,u,a_0} : u \in [0,t]\}$. When conditioning on the training set $\s{T}_{n,k}$, the functions $S_{n,k}$, $G_{n,k}$, $\pi_{n,k}$ and $\Lambda_{n,k}$ are fixed, so Lemma~\ref{lemma:if_entropy} implies that
 \[ \log \sup_Q N(\varepsilon \|\bar{G}_{n,k,t,a_0}\|_{Q,2},  \s{G}_{n,k,t,a_0} , L_2(Q)) \leq \tilde{C} \log \varepsilon^{-1}\]
 for some constant $\tilde{C}$ not depending on $n$, $k$, or $\varepsilon$, and where $\bar{G}_{n,k,t,a_0} := \sup_{u \in [0,t]} |\phi_{n,k,u,a_0} - \phi_{\infty,u,a_0}|$ is the natural envelope function for $\s{G}_{n,k,t,a_0}$. As a result, the uniform entropy integral 
 \[ J(1, \s{G}_{n,k,t,a_0} , L_2(P_0)) = \sup_Q \int_0^1 \left[ 1 + \log N(\varepsilon \|\bar{G}_{n,k,t,a_0}\|_{Q,2},  \s{G}_{n,k,t,a_0} , L_2(Q))\right]^{1/2} \, d\varepsilon  \]
 is bounded by a constant not depending on $n$ or $k$. By Theorem~2.14.2 of \cite{van1996weak}, there is therefore a constant $\bar{C}$  not depending on $n$ or $k$ such that
 \begin{align*}
 E_0\left[ E_0 \left[ \sup_{g \in \s{G}_{n,k,t,a_0}}  | \d{G}_{n}^k g  | \,\middle|\, \s{T}_{n,k} \right] \right]&\leq \bar{C} E_0\left[ E_0 \left[ \sup_{u \in [0,t]} \left[\phi_{n,k,t,a_0}(O) - \phi_{\infty,t,a_0}(O)\right]^2 \,\middle|\, \s{T}_{n,k} \right]^{1/2} \right] \\
 &\leq  \bar{C} \left\{ E_0\left[  \sup_{u \in [0,t]} \left[\phi_{n,k,t,a_0}(O) - \phi_{\infty,t,a_0}(O)\right\}^2 \right] \right\}^{1/2}.
\end{align*}
By Lemma~\ref{lemma:influence_bound}, we therefore have that $
\frac{1}{K} \sum_{k=1}^K   \sup_{u \in [0,t]} \left| \d{G}_{n}^k \left(\phi_{n,k,u,a_0} - \phi_{\infty,u,a_0}\right)  \right|$ is bounded above by $C\bar{C}\max_k C^*_k$, where $C^*_k:=C^*_{1k}+C_{2k}+C_{3k}$ with
\begin{align*}
C_{1k}^*:= E_{0}\left[ \sup_{u \in [0,t]}\sup_{v \in [0,u]} \left| \frac{S_{n,k}(u \,|\, a_0, W)}{S_{n,k}(v \,|\, a_0, W)} - \frac{S_{\infty}(u \,|\, a_0, W)}{S_{\infty}(v \,|\, a_0, W)} \right |  \right]^2.
\end{align*}
By~\ref{itm:conv} and~\ref{itm:unif_conv}, this tends in probability to zero.
 
\end{proof}

\begin{proof}[\bfseries{Proof of Theorem~\ref{thm:consistency}}]
By Lemma~\ref{lemma:expansion} and the triangle inequality, \ref{itm:dr} implies that $\left|\theta_{n}(t,a_0) -\theta_0(t,a_0) \right|$ is bounded above by
\begin{align*}
\left| \d{P}_n \phi_{\infty,t,a_0}^* \right|+ \left|\frac{1}{K} \sum_{k=1}^K \frac{K n_k^{1/2}}{n} \d{G}_{n}^k \left(\phi_{n,k,t,a_0} - \phi_{\infty,t,a_0}\right)\right|+ \left|\frac{1}{K} \sum_{k=1}^K \frac{K n_k}{n} P_0 \left(\phi_{n,k,t,a_0} - \phi_{\infty, t,a_0}\right) \right|.
\end{align*}
Since the first term is an empirical mean of a mean zero function by Lemma~\ref{lemma:expansion}, it is $\fasterthan(1)$ by the weak law of large numbers. By Lemma~\ref{lemma:emp_process}, \ref{itm:bound} and~\ref{itm:conv} imply that the second term is $\fasterthan(n^{-1/2})$. We note that 
\[\frac{Kn_k}{n} = \frac{K | n_k - n/K + n/K|}{n} \leq \frac{K |n_k - n/K| + n}{n} \leq K/n + 1 \leq 2\] 
since $ |n_k - n/K| \leq 1$. Therefore, by the triangle and Cauchy-Schwarz inequalities, the third term is bounded by 
\[ 2  \left[ \max_k P_0\left(  \phi_{n,k,t,a_0} - \phi_{\infty, t,a_0}\right)^2 \right]^{1/2}.\]
By Lemma~\ref{lemma:influence_bound}, \ref{itm:bound} implies that this is bounded by $
2C \left(\max_k C_{1k}\right)^{1/2}+2C\left(\max_k C_{2k}\right)^{1/2}+2C\left(\max_k C_{3k}\right)^{1/2}$
for $C$ depending only on $\eta$, and $C_{1k}$, $C_{2k}$ and $C_{3k}$ defined as in the proof of Lemma 6. By \ref{itm:conv}, this upper bound is $\fasterthan(1)$, which implies that $\left|\theta_{n}(t,a_0) -\theta_0(t,a_0) \right| = \fasterthan(1)$.

For uniform consistency, Lemma~\ref{lemma:expansion} and the triangle inequality, \ref{itm:dr} implies that
\begin{align*}
\sup_{u\in[0,t]} \left|\theta_{n}(u,a_0) -\theta_0(u,a_0) \right| &\leq \sup_{u\in[0,t]}  \left| \d{P}_n \phi_{\infty,u,a_0}^* \right|+ \sup_{u\in[0,t]}\left|\frac{1}{K} \sum_{k=1}^K \frac{K n_k^{1/2}}{n} \d{G}_{n}^k \left(\phi_{n,k,u,a_0} - \phi_{\infty,u,a_0}\right)\right| \\
&\qquad+ \sup_{u\in[0,t]}\left|\frac{1}{K} \sum_{k=1}^K \frac{K n_k}{n} P_0 \left(\phi_{n,k,u,a_0} - \phi_{\infty, u, a_0}\right) \right|.
\end{align*}
Lemma~\ref{lemma:if_entropy} implies that $\{\phi_{\infty, u, a_0} : u \in [0,t]\}$ is a $P_0$-Donsker class, so the first term on the right-hand side of the inequality above is $\bounded(n^{-1/2})$. By Lemma~\ref{lemma:emp_process}, \ref{itm:bound} and~\ref{itm:unif_conv} imply that the second term  is $\fasterthan(n^{-1/2})$. As above, the third term is $\fasterthan(1)$ by Lemma~\ref{lemma:influence_bound}. We thus find that $\sup_{u \in [0,t]}\left|\theta_{n}(u,a_0) -\theta_0(u,a_0) \right| = \fasterthan(1)$.
\end{proof}

\begin{proof}[\bfseries{Proof of Theorem~\ref{thm:asymptotic_linearity}}]
Since \ref{itm:dr} holds automatically when $S_\infty = S_0$, $G_\infty = G_0$ and $\pi_\infty = \pi_0$, Lemma~\ref{lemma:expansion} implies that 
\begin{align}
\theta_{n}(t,a_0) -\theta_0(t,a_0) &= \d{P}_n \phi_{0,t,a_0}^* + \frac{1}{K} \sum_{k=1}^K \frac{K n_k^{1/2}}{n} \d{G}_{n}^k \left(\phi_{n,k,t,a_0} - \phi_{0,t,a_0}\right) + \frac{1}{K} \sum_{k=1}^K \frac{K n_k}{n} \left[P_0 \phi_{n,k,t,a_0} - \theta_0(t,a_0)\right]\label{eq:theta_decomp_asy_lin}.
\end{align}
Since~\ref{itm:bound} and~\ref{itm:conv} hold by assumption, the second summand on the right-hand side is $\fasterthan(n^{-1/2})$ by Lemma~\ref{lemma:emp_process}, where we replace the symbol $\infty$ by $0$ throughout.  By Lemma~\ref{lemma:if_mean}, $P_0 \phi_{n,k,t,a_0} - \theta_0(t,a_0)$ equals 
\[E_{0} \left[ S_{n,k}(t \,|\, a_0, W) \int_0^t \frac{S_0(u\m \,|\, a_0, W)}{S_{n,k}(u \,|\, a_0, W)}\left\{\frac{\pi_0(a_0 \,|\, W) G_0(u \,|\, a_0, W)}{\pi_{n,k}(a_0 \,|\, W)G_{n,k}(u \,|\, a_0, W)} - 1\right\} (\Lambda_{n,k} - \Lambda_0)(du \,|\, a_0, W)   \right].\]
By the Duhamel equation (Theorem 6 of \citealp{gill1990product}), we have that
\[\frac{S_0(u\m \,|\, a_0, W)}{S_{n,k}(u \,|\, a_0, W)}(\Lambda_{n,k} - \Lambda_0)(du \,|\, a_0, W) =\left( \frac{S_0}{S_{n,k}} - 1 \right)(du \,|\, a_0, W),\]
and so the above equals
\begin{align*}
&E_{0} \left[ S_{n,k}(t \,|\, a_0, W) \int_0^t \left\{\frac{\pi_0(a_0 \,|\, W) G_0(u \,|\, a_0, W)}{\pi_{n,k}(a_0 \,|\, W)G_{n,k}(u \,|\, a_0, W)} - 1\right\} \left(\frac{S_0}{S_{n,k}} - 1 \right)(du \,|\, a_0, W)   \right] \\
&\qquad=E_{0} \left[ S_{n,k}(t \,|\, a_0, W) \int_0^t \left\{ \frac{\pi_0(a_0 \,|\, W) }{\pi_{n,k}(a_0 \,|\, W)} - 1\right\} \left(\frac{S_0}{S_{n,k}} - 1 \right)(du \,|\, a_0, W)   \right]\\
&\qquad\qquad + E_{0} \left[ S_{n,k}(t \,|\, a_0, W) \int_0^t \frac{\pi_0(a_0 \,|\, W) }{\pi_{n,k}(a_0 \,|\, W)} \left\{ \frac{G_0(u \,|\, a_0, W)}{G_{n,k}(u \,|\, a_0, W)}  - 1 \right\} \left(\frac{S_0}{S_{n,k}} - 1 \right)(du \,|\, a_0, W)   \right]\\
&\qquad=E_{0} \left[\left\{ \frac{\pi_0(a_0 \,|\, W) }{\pi_{n,k}(a_0 \,|\, W)} - 1\right\} S_{n,k}(t \,|\, a_0, W) \int_0^t\frac{S_0}{S_{n,k}}(du \,|\, a_0, W) \right] \\
&\qquad\qquad+E_0\left[\frac{\pi_0(a_0 \,|\, W) }{\pi_{n,k}(a_0 \,|\, W)}S_{n,k}(t \,|\, a_0, W) \int_0^t \left\{ \frac{G_0(u \,|\, a_0, W)}{G_{n,k}(u \,|\, a_0, W)}  - 1 \right\} \left( \frac{S_0}{S_{n,k}} - 1\right)(du \,|\, a_0, W)   \right] \\
&\qquad=E_{0} \left[ \frac{\left\{ \pi_{n,k}(a_0 \,|\, W) -  \pi_0(a_0 \,|\, W) \right\} \left\{ S_{n,k}(t \,|\, a_0, W) - S_0(t \,|\, a_0, W)\right\}}{\pi_{n,k}(a_0 \,|\, W)}\right] \\
&\qquad\qquad+E_{0} \left[ \frac{\pi_0(a_0 \,|\, W) }{\pi_{n,k}(a_0 \,|\, W)}S_{n,k}(t \,|\, a_0, W) \int_0^t \left\{ \frac{G_0(u \,|\, a_0, W)}{G_{n,k}(u \,|\, a_0, W)}  - 1 \right\} \left( \frac{S_0}{S_{n,k}} - 1\right)(du \,|\, a_0, W)   \right].
\end{align*}
In view of \ref{itm:bound}, the absolute value of the right-hand side of the last equality is bounded above by
\begin{align*}
 &\eta E_{0} \left[ \left| \pi_{n,k}(a_0 \,|\, W) -  \pi_0(a_0 \,|\, W) \right| \left| S_{n,k}(t \,|\, a_0, W) - S_0(t \,|\, a_0, W)\right|\right] \\
&\qquad+\eta E_0 \left[ \left| S_{n,k}(t \,|\, a_0, W) \int_0^t \left[ \frac{G_0(u\,|\, a_0, W)}{G_{n,k}(u \,|\, a_0, W)}  - 1 \right]  \left(\frac{S_0}{S_{n,k}} -1 \right)(du \,|\, a_0, W)  \right| \right].
\end{align*}
The maximum over $k$ of these expressions equals $\eta (r_{n,t,a,1} + r_{n,t,a,2})$. Therefore, since $\frac{1}{K} \sum_{k=1}^K \frac{K n_k}{n} \leq 2$ (as discussed in the proof of Theorem~\ref{thm:consistency} above), we have
\[ \left|\frac{1}{K} \sum_{k=1}^K \frac{K n_k}{n} \left[P_0 \phi_{n,k,t,a_0} - \theta_0(t,a_0)\right]  \right| \leq 2\eta(r_{n,t,a_0,1} + r_{n,t,a_0,2})\ ,\]which 
 is $\fasterthan(n^{-1/2})$ by \ref{itm:pointwise_rate}. This establishes that $\theta_n(t,a_0) = \theta_0(t,a_0) + \d{P}_n \phi_{0,t,a_0}^* + \fasterthan(n^{-1/2})$, as claimed in Theorem~\ref{thm:asymptotic_linearity}. Since $\phi_{0,t,a_0}^*$ is uniformly bounded, $P_0\phi^{*2}_{0,t,a_0}<\infty$, and since $P_0\phi^*_{0,t,a_0}=0$, it follows, as claimed, that 
\[n^{1/2} \d{P}_n \phi_{0,t,a_0}^* \indist N\left( 0, P_0\phi_{0,t,a_0}^{*2}\right).\]

For the uniform statements, we use the same decomposition as in~\eqref{eq:theta_decomp_asy_lin}. Since~\ref{itm:bound},~\ref{itm:conv} and~\ref{itm:unif_conv} hold by assumption,
\[\sup_{u \in [0,t]} \left|  \frac{1}{K} \sum_{k=1}^K \frac{K n_k^{1/2}}{n} \d{G}_{n}^k \left(\phi_{n,k,u,a_0} - \phi_{0,u,a_0}\right)\right| = \fasterthan(n^{-1/2})\] 
by Lemma~\ref{lemma:emp_process}, where we replace the symbol $\infty$ by $0$ throughout. Following the derivation above, we have that
\[ \sup_{u \in [0,t]} \left| \frac{1}{K} \sum_{k=1}^K \frac{K n_k}{n} \left[P_0 \phi_{n,k,u,a_0} - \theta_0(u,a_0)\right]\right| \leq \sup_{u \in [0,t]} 2\eta(r_{n,u,a_0,1} + r_{n,u,a_0,2}),\]
which is $\fasterthan(n^{-1/2})$ by~\ref{itm:unif_rate}. Therefore, $ \sup_{u \in [0,t]}\left| \theta_n(u,a_0) - \theta_0(u,a_0) - \d{P}_n \phi_{0,u,a_0}^*\right| = \fasterthan(n^{-1/2})$
as claimed. Since $\{\phi_{0,u,a_0}^* : u \in [0,t]\}$ is a uniformly bounded $P_0$-Donsker class by Lemma~\ref{lemma:if_entropy}, $\left\{n^{1/2}  \d{P}_n \phi_{0,u,a_0}^* : u \in [0, t]\right\}$ converges weakly to a tight mean-zero Gaussian process with covariance $(u,v) \mapsto P_0\left( \phi_{0,u,a_0}^*\phi_{0,v,a_0}^*\right)$.
\end{proof}

\begin{proof}[\bfseries{Proof of Theorem~\ref{thm:loss}}]
We first note that for $P_0$-almost every $(a,w)$ and all $t \leq \tau$, conditions (A1)--(A5) imply that
\begin{align*}
E_0\left[ \frac{\Delta I(Y \leq t)}{G_0(Y \,|\, a, w)}\, \middle|\, A = a, W =w\right] &= -\int_0^\infty \frac{I(y \leq t) }{G_0(y \,|\, a, w)}G_0(y  \,|\, a, w) S_0(dy \,|\, a, w) = -\int_0^{t}  S_0(dy \,|\, a, w) \\
&= 1 - S_0(t  \,|\, a, w) = P_{0,F}(T \leq t \, | \, A = a, W=w)\ .
\end{align*}
Therefore, for any such $(t, a,w)$
 \[ s \mapsto  E_0 \left[ s \left[s - 2 \left\{ 1 - \frac{\Delta I(Y \leq t)}{G_0(Y \,|\, A, W)}\right\}\right]\,\middle|\, A = a, W =w\right] =s \left\{s - 2 S_0(t \, | \, a,w)\right\} \]
 is uniquely minimized by $s = S_0(t \, | \, a,w)$. Hence, any minimizer $S^*$ of
 \[ S \mapsto E_0\left[ \int_0^\tau E_0 \left[ S(t \, | \, A,W) \left[S(t \, | \, A,W) - 2 \left\{ 1 - \frac{\Delta I(Y \leq t)}{G_0(Y \,|\, A, W)}\right\}\right]\,\middle|\, A , W\right] dt\right]\]
satisfies $S^*(t \, | \, a,w) = S_0(t \, | \, a,w)$ for $P_0$-almost every $(a,w)$ and all $t \leq \tau$. If the integral with respect to $t$ can be exchanged with the conditional expectation with respect to $(Y,\Delta)$ given $A=a$ and $W=w$, then the result follows. To justify exchanging these integrals, we use Fubini's theorem, which requires demonstrating that
 \[ \int_0^\tau E_0\left\{ \left| S(t \, | \, a,w) \left[S(t \, | \, a,w) - 2 \left\{ 1 - \frac{\Delta I(Y \leq t)}{G_0(Y \,|\, a, w)}\right\}\right] \right| \, \middle|\,A = a, W= w\right\}dt< \infty\]
 for $P_0$-almost every $(a,w)$. Since $S(t \, | \, a,w)$ and $G_0(t \, | \, a,w)$ are both contained in $[0,1]$, it holds that
 \begin{align*}
 &\left| S(t \, | \, a,w) - 2 \left\{ 1 - \frac{\delta I(y \leq t)}{G_0(y\,|\, a, w)}\right\}\right|  \\
 &=  \delta I(y\leq t) \left| S(t \, | \, a,w) - 2\left\{1 - \frac{1}{G_0(y\,|\, a, w)}\right\} \right| + [1-\delta I(y \leq t)] | S(t \, | \, a,w) - 2| \\
 &= \delta I(y\leq t) \left\{ S(t \, | \, a,w) - 2 + \frac{2}{G_0(y\,|\, a, w)}\right\} + [1-\delta I(y \leq t)] [ 2- S(t \, | \, a,w)]\leq 2\left\{ 1+  \frac{\delta I(y\leq \tau)}{G_0(y\,|\, a, w)}\right\}\ .
 \end{align*}
 Thus, we find that
 \begin{align*}
& \int_0^\tau E_0\left[ \left| S(t \, | \, a,w) \left[S(t \, | \, a,w) - 2 \left\{ 1 - \frac{\Delta I(Y \leq t)}{G_0(Y \,|\, a, w)}\right\}\right] \right| \, \middle|\, A = a, W = w\right] dt \\
&\qquad\leq 2\tau + 2 \int_0^\tau E_0\left[\frac{\Delta I(Y\leq \tau)}{ G_0(Y\,|\, a, w)} \,\middle|\,A = a, W = w\right] dt = 2\tau  - 2\tau \int_0^\tau S_0(dy \, | \, a, w) \leq 4\tau
 \end{align*}
 for $P_0$-almost every $(a,w)$. Hence, the order of integration may be exchanged.

Similarly, for any $(a, w)$ in the support of $(A,W)$ and $t \leq \tau$ such that $S_0(t\m \, | \, a, w) > 0$, we note that
\begin{align*}
E_0\left[ \frac{(1-\Delta)I(Y < t) }{S_0(Y \,|\, A, W)}\,\middle|\, A = a, W =w\right] &= -\int_0^\infty \frac{I(y < t)}{S_0(y \,|\, a, w)} S_0(y  \,|\, a, w) G_0^+(dy \,|\, a, w) \\
&= -\int_{[0,t)}  G_0^+(dy \,|\, a, w) = 1 - G_0(t  \,|\, a, w)\\
&= P_{0,F}(C < t \, | \, A = a, W=w)\ ,
\end{align*}
where $G_0^+$ is the right-continuous version of $G_0$ (which is left-continuous by definition). Therefore, for any such $(t, a,w)$, it follows that
 \[ g \mapsto  E_0 \left[ g \left[g - 2 \left\{ 1 - \frac{(1-\Delta) I(Y < t)}{S_0(Y \,|\, A, W)}\right\}\right]\,\middle|\, A = a, W =w\right] =g \left\{g - 2 G_0(t \, | \, a,w)\right\} \]
 is uniquely minimized by $g = G_0(t \, | \, a,w)$. For $t \leq \tau$ such that $S_0(t\m \, | \, a, w)= 0$, we have that
 \begin{align*}
E_0\left[ \frac{(1-\Delta)I(Y < t) }{S_0(Y \,|\, A, W)}\,\middle|\, A = a, W =w\right] &= -\int_0^\infty \frac{I(y < t)}{S_0(y \,|\, a, w)} S_0(y  \,|\, a, w)\,G_0^+(dy \,|\, a, w) \\
&= -\int_{[0,t^+(a,w))}  G_0^+(dy \,|\, a, w) = 1 - G_0(t^+(a,w)  \,|\, a, w)\ ,
\end{align*}
 where $t^+(a,w) := \inf\{t : S_0(t\, | \, a, w)= 0\}$. Hence, it follows that
 \[ G \mapsto E_0\left[ \int_0^\tau E_0 \left\{ G(t \, | \, A,W) \left[G(t \, | \, A,W) - 2 \left\{ 1 - \frac{(1-\Delta) I(Y < t)}{S_0(Y \,|\, A, W)}\right\}\right]\,\middle|\, A , W\right\} \,dt \right] \]
 is minimized by $G^*$, where $G^*(t \, | \, a,w) = G_0(t \, | \, a,w)$ if $S_0(t\m \, | \, a, w) > 0$  and $G^*(t \, | \, a,w) = G_0(t^+(a,w) \, | \, a,w)$ otherwise for $P_0$-almost every $(a,w)$. Similarly as was done above, it is straightforward to show that the inner function is integrable, thus justifying the use of Fubini's theorem for exchanging the order of integration. This completes the proof.

\end{proof}

\end{adjustwidth}

\end{document}